\newif{\ifdraft}\drafttrue
\draftfalse
\ifdraft
\documentclass[draft,oribibl]{llncs}
\usepackage[notcite,color]{showkeys}
\else
\documentclass[oribibl,final]{llncs}
\fi 
\usepackage{etex}
\usepackage[applemac]{inputenc}
\usepackage{amsmath}
\usepackage{amssymb}
\usepackage{theorem}
\usepackage{lmodern}
\usepackage{fontenc}[T1]
\usepackage{enumerate}
\usepackage{enumitem}
\usepackage[english]{babel}
\usepackage{ifthen}

\usepackage{multirow}
\usepackage{xspace}
\usepackage{verbatim}

\usepackage[ruled]{algorithm}
\usepackage[noend]{algpseudocode}

\usepackage[final]{hyperref}

\usepackage{pdfsync}

\DeclareMathOperator{\End}{End}

\usepackage[ruled]{algorithm}
\usepackage[noend]{algpseudocode}

\usepackage{tikz}
\usetikzlibrary{arrows,decorations.markings}

\definecolor{darkgreen}{rgb}{0.0,0.7,0.0}
\newenvironment{LC}{\noindent\color{darkgreen} LC:}{}

\newenvironment{vd}{\noindent\color{blue} \colorbox{blue}{\color{black} VD:}}{}

\newenvironment{ME}{\noindent\color{magenta} ME:} {}

\usepackage{prettyref}
\newcommand{\prref}[1]{\prettyref{#1}}
\newrefformat{thm}{Theorem~\ref{#1}}
\newrefformat{lem}{Lemma~\ref{#1}}
\newrefformat{def}{Definition~\ref{#1}}
\newrefformat{cor}{Corollary~\ref{#1}}
\newrefformat{prop}{Proposition~\ref{#1}}
\newrefformat{sec}{Section~\ref{#1}}
\newrefformat{subsec}{Subsection~\ref{#1}}
\newrefformat{kap}{Chapter~\ref{#1}}
\newrefformat{ex}{Example~\ref{#1}}
\newrefformat{rem}{Remark~\ref{#1}}
\newrefformat{fig}{Figure~\ref{#1}}
\newrefformat{app}{Appendix~\ref{#1}}
\newrefformat{eq}{Equation~(\ref{#1})}
\newrefformat{tab}{Table~\ref{#1}}
\newrefformat{ap}{{\sc Appendix}}

\theoremstyle{plain}

\newcommand{\ie}{i.e.,\xspace}
\newcommand{\eg}{e.g.\xspace}

\newcommand{\subst}{substitution\xspace}

\newcommand{\Dio}{Diophantine\xspace}

\newcommand{\solu}{solution\xspace}
\newcommand{\invol}{involution\xspace}

\newcommand{\ql}{quasi-linear\xspace}
\newcommand{\qls}{quasi-linear space\xspace}

\newcommand{\IFF}{if and only if\xspace}
\renewcommand{\hom}{homomorphism\xspace}
\newcommand{\homs}{homomorphisms\xspace}
\newcommand{\Endo}{endomorphism\xspace}
\newcommand{\Endos}{endomorphisms\xspace}

\newcommand{\morph}{morphism\xspace}

\newcommand{\lds}{, \ldots ,}

\newcommand{\ra}{\longrightarrow}

\newcommand{\df}[2][]{\textbf{\color{black}\sffamily #2}}

\usepackage[framemethod=tikz]{mdframed} 

\newcommand{\Apos}{A_{+}}
\newcommand{\Aneg}{A_{-}}
\newcommand{\Apone}{A_{\pm}}
\newcommand{\Bpos}{B_{+}}
\newcommand{\Bneg}{B_{-}}

\newcommand{\AF}{A_{\F}}

\newcommand{\FG}[1]{\text{F}({#1})}

\newcommand{\ninit}{n_{\mathrm{init}}}
\newcommand{\muinit}{\mu_{\mathrm{init}}}
\newcommand{\Winit}{W_{\mathrm{init}}}

\newcommand{\arc}[1]{\overset{#1}\ra}
\newcommand{\set}[2]{\left\{#1\mathrel{\left|\vphantom{#1}\vphantom{#2}\right.}#2\right\}}

\newcommand{\oneset}[1]{\left\{\mathinner{#1}\right\}}
\newcommand{\os}{\oneset}
\newcommand{\sm}{\setminus}
\newcommand{\es}{\emptyset}
\newcommand{\sse}{\subseteq}

\newcommand{\vdmatrix}[4]{\left(\begin{smallmatrix}#1 & #2\\ #3 & #4\end{smallmatrix}\right)}

\newcommand{\abs}[1]{\left|\mathinner{#1}\right|}
\newcommand{\Abs}[1]{\left\Vert\mathinner{#1}\right\Vert}

\newcommand{\N}{\ensuremath{\mathbb{N}}}
\newcommand{\Z}{\ensuremath{\mathbb{Z}}}

\newcommand{\R}{\ensuremath{\mathbb{R}}}

\newcommand{\B}{\ensuremath{\mathbb{B}}}
\newcommand{\F}{\ensuremath{\mathbb{F}}}

\newcommand{\NP}{\ensuremath{\mathsf{NP}}}
\newcommand{\NSPACE}{\ensuremath{\mathsf{NSPACE}}}

\newcommand{\Bn}{\B^{n\times n}}

\renewcommand{\phi}{\varphi}
\newcommand{\eps}{\varepsilon}

\newcommand{\alp}{\alpha}
\newcommand{\bet}{\beta}

\newcommand{\del}{\delta}
\newcommand{\lam}{\lambda}
\newcommand{\sig}{\sigma}

\newcommand{\Sig}{\Sigma}
\newcommand{\Gam}{\GG}

\newcommand\GG{\Gamma}

\newcommand\Lam{\Lambda}
\newcommand\OO{\Omega}

\newcommand\SL{\mathop\mathrm{SL}}
\newcommand\PSL{\mathop\mathrm{PSL}}

\newcommand{\Oh}{\mathcal{O}}

\newcommand{\wt}[1]{\widetilde{ #1 }}

\newcommand{\id}[1]{\mathrm{id}_{#1}}
\newcommand{\cA}{\mathcal{A}}

\newcommand{\cG}{\mathcal{G}}
\newcommand{\cL}{\mathcal{L}}

\newcommand{\cX}{\mathcal{X}}

\newcommand{\cR}{\mathcal{R}}

\newcommand{\cSol}{\mathrm{Sol}}

\newcommand{\ov}[1]{\overline{#1}}

\newcommand{\oi}[1]{{#1}^{-1}}
\newcommand{\mathinvol}{\overline{\,^{\,^{\,}}}}

\newcommand{\Prob}[1]{\mathrm{Pr}\left[\, #1 \,\right]}

\newcommand{\Rat}{\mathrm{RAT}}

\newcommand{\bcc}{29n} 
\newcommand{\pcc}{31n} 

\newif{\ifsecappendix}\secappendixtrue
\newif{\ifsecquestion}\secquestiontrue
\newif{\ifsecoldstuff}\secoldstufftrue
\newif{\ifAlles}\Allestrue

\Allesfalse 
\ifAlles\else
\secappendixfalse 
\secquestionfalse  
\secoldstufffalse  
\fi

\begin{document}

\pagestyle{plain}

\title{Solution sets for
equations over free groups  are {EDT0L} languages --- ICALP 2015 version\thanks{Research supported by the
 Australian Research Council FT110100178 and the University of Newcastle G1301377.
The first author was supported by a Swiss National Science Foundation Professorship FN PP00P2-144681/1. The first and third authors were supported by a University of Neuch\^atel Overhead grant in 2013.}}
\author{%
Laura Ciobanu\inst{1},
  Volker Diekert\inst{2}, \and
  Murray Elder\inst{3} 
   }
\authorrunning{L.{} Ciobanu\\
  V.{} Diekert\inst{2} \and
  M.{} Elder}

\institute{%
Institut de math\'ematiques, Universit{\'e} de Neuch{\^a}tel, Switzerland \and 
  Institut f\"ur Formale Methoden der Informatik,
  Universit\"at Stuttgart, Germany \and 
School of Mathematical \& Physical Sciences, The University of Newcastle,
Australia}

\maketitle


\begin{abstract}We show that, given a word equation over a finitely generated free group, the set of all solutions in reduced words forms an EDT0L language. In particular, it is an  indexed language in the sense of Aho. The question of whether a description of solution sets in reduced words as an indexed language is possible has been been open for some years \cite{GilPC12,JainMS2012Lics}, apparently without much hope that 
 a positive answer could hold. Nevertheless, our  answer goes far beyond: they are EDT0L, which is a proper subclass of indexed languages. We can additionally handle the existential theory of equations with rational constraints in free products $\star_{1 \leq i \leq s}F_i$, where each $F_i$ is either a free or finite group, or a free monoid with involution.
In all cases the result is the same: the set of all solutions in reduced words is EDT0L.
This was known only for quadratic word equations by \cite{FerteMarinSenizerguesTocs14}, which is a very restricted case. Our general result became possible due to the recent recompression technique of Je\.z. In this paper we 
use a new method to integrate
solutions of linear \Dio equations into the process and obtain more general results than in the related paper \cite{DiekertJP2014csr}.
For example, we improve the complexity from quadratic nondeterministic space in \cite{DiekertJP2014csr} to quasi-linear nondeterministic space here. 
This implies an improved complexity for deciding the existential theory of non-abelian free groups: $\NSPACE(n\log n$). The conjectured complexity is \NP, however, we believe that our results are optimal 
with respect to 
 space complexity, independent of the conjectured \NP. 
 \end{abstract}

\vspace{-0.8cm}
\subsection*{Introduction and main results}\label{sec:intro}
The first algorithmic description of all solutions to a given equation over a free group is due to Razborov \cite{raz87,raz93}. 
His description became known as a \emph{Makanin-Razborov diagram}. This concept plays a major role in the positive solution of Tarski's conjectures about the elementary theory in free groups 
\cite{KMIV06,sela13}. 

It was however unknown that there is an amazingly simple formal language 
description for the set of all solutions of an equation over free groups 
in reduced words: they are EDT0L. An EDT0L language $L$ is given by a nondeterministic finite automaton (NFA), where transitions are labeled by \Endos in a free monoid 
 which contains a symbol $\#$. 
Such an NFA defines a rational language $\cR$ of \Endos, and 
the condition on $L$ is that $L=\set{h(\#)}{h\in \cR}$. 
The NFA we need for our result 
can be computed effectively in nondeterministic quasi-linear space, \ie by some $\NSPACE(n\log n$) algorithm. As a consequence, the automaton has singly exponential size $2^{\Oh(n\log n)}$ in the input size $n$. 

A description of solution sets as EDT0L languages
was known before only for quadratic word equations by \cite{FerteMarinSenizerguesTocs14}; the recent paper 
\cite{DiekertJP2014csr} did not aim at giving such a structural result. 
There is also a description of all solutions for a word equation by Plandowski in~\cite{Plandowski06stoc}.
His description is given by some graph which can computed in singly exponential time, but without the aim to give any formal language characterization. Plandowski claimed in~\cite{Plandowski06stoc} that his method applies also to free groups with rational constraints, but he found a gap~\cite{Pl6}.

The technical results are as follows.  
Let $\FG \Apos$ be the free group over a finite generating set $\Apos$
of (positive) letters. We let 
$\Apone = \Apos \cup \set{\oi a}{a \in \Apos}\sse \FG \Apos$. We view $\Apone$ as a finite alphabet (of {\em constants}) with the involution 
$\ov a = \oi a$. The involution is extended to the free monoid 
$\Apone^*$  by $\ov{a_{1}\cdots a_{k}} = \ov{a_{k}}\cdots \ov{a_{1}}$. We let $\pi:\Apone^* \to \FG \Apos$ be the canonical \morph. 
As a set, we identify $\FG \Apos$ with the rational (\ie regular) subset of reduced words inside $\Apone^*$. A word  is  {\em reduced} if it does not contain any factor $a \ov a$ where $a \in \Apone$. 
Thus, $w\in \Apone^*$ is reduced \IFF $\pi(w) = w$. We emphasize that $\FG \Apos$ is realized as a subset of $\Apone^*$. Let $\OO$ be a set of {\em variables} with involution. 
An {\em equation} over $\FG \Apos$ is given as a pair 
$(U,V)$, where $U, V\in (\Apone\cup \OO)^*$ are words over constants and variables. A {\em solution} of $(U,V)$ is a mapping
$\sig: \OO \to \Apone^*$ which respects the involution such that
$\pi \sig(U) = \pi \sig(V)$ holds
in $\FG \Apos$.
As usual, $\sig$ is extended to a \morph $\sig: (\Apone\cup \OO)^* \to \Apone^*$
by leaving constants invariant. 
Throughout we let $\#$ denote a special symbol, whose main purpose is to 
encode a tuple of words $(w_{1}, \ldots, w_{k})$ as a single word 
$w_{1}\# \cdots\# w_{k}$.

\begin{theorem}\label{thm:freecentral}
Let  $(U,V)$ be an equation over $\FG \Apos$ and  
$\os{X_1,\dots, X_k}$  be any specified subset of variables.
Then the solution set  
$\cSol(U,V)$ is EDT0L where 
$\cSol(U,V)= \set{\sig(X_1)\# \cdots \#\sig(X_k)}{\sig \text { solves $(U,V)$ in reduced words}}$. 

Moreover, there is a nondeterministic  algorithm which takes $(U,V)$ as input and computes an NFA $\cA$ 
 such that $\cSol(U,V)= \set{\phi(\#)}{\phi \in L(\cA)}$ 
 in quasi-linear space. 
 \end{theorem}
 The statement of \prref{thm:freecentral} shifts the perspective on how to solve equations. 
 Instead of solving an equation, we focus on  an effective construction of some NFA
 producing the EDT0L set. Once the NFA is constructed, the existence of a solution, or whether the number of solutions is zero, finite or infinite, 
become graph properties of the NFA.

\prref{thm:freecentral}
is a special case of a more general result
involving the existential theory with rational constraints over free products. 
The generalization is done in several directions. 
First, we can replace $\FG \Apos$ by any finitely generated free product $\F=\star_{1 \leq i \leq s}F_i$ where each $F_i$ is either a free or finite group, or a free monoid with arbitrary involutions (including the identity).  %
Thus, for example we may have $\F= \os{a,b}^* \star \Z \star\PSL(2,\Z)
= \os{a,b}^* \star \Z \star (\Z/3\Z)\star (\Z/2\Z)$
where $\ov a = a$ and $\ov b = b$.  Second,  we allow arbitrary rational constraints. We consider Boolean 
formulae $\Phi$, where each atomic formula is either an equation or a rational constraint, written as
$X \in L$, where $L\sse \F$ is a rational subset.

Allowing rational constraints makes it necessary to specify how the input for a constraint is given. 
We do so algebraically, by using a \morph $\rho: A^* \to N$, where $N$ is a finite monoid with involution
and $A = \oi A\sse \F$ generates $\F$. 
Thus, we write a constraint  in the form $X \in \oi{\rho}(m)$, with $m \in N$,
and the interpretation $\rho \sig(X) = m$.  The input size $\Abs \Phi$ is given by the sum of the lengths 
of all atomic formulae, together with $(\abs  A +  \abs \OO)(1 +\log \abs N)$. 
The specification that the solution is in reduced words increases the input size by at most a factor of $\Oh(\log \abs A)$.

\begin{theorem}\label{thm:procentral}
Let $\F$ be a free product as above, $\Phi$ a Boolean formula over equations and rational constraints, and
$
\os{X_1,\dots, X_k}$ any subset of variables.
Then $\cSol(\Phi)= \set{\sig(X_1)\# \cdots \#\sig(X_k)}{\sig \ \text {solves $\Phi$ in reduced words}}$
 is  EDT0L. 
 
 Moreover, there is an algorithm which takes 
 $\Phi$ as input and produces an NFA $\cA$ 
 such that $\cSol(\Phi)= \set{\phi(\#)}{\phi \in L(\cA)}$. 
The algorithm is nondeterministic and  uses  quasi-linear 
space in the input size   $\Abs \Phi$.
\end{theorem}

The proof of \prref{thm:freecentral} is given in \prref{sec:pr}.
Part II of the paper is devoted to the proof of \prref{thm:procentral}
which is more technical and more difficult.
\section{Preliminaries}\label{sec:EDT0L}
We assume that the reader is familiar with big-$\Oh$ and big-$\Theta$ notation. A function $f:\N \to \R$ is called {\em quasi-linear} if we have 
$\abs{f(n)} \in \Oh(n \log n)$. 
For results and notation in complexity theory we refer to the textbook \cite{pap94}.
We also use standard notation from combinatorics on words and automata theory
according to \cite{eil74}.

\subsection{Words and involutions}\label{sec:rational}
If $A$ is a set then $A^*$ denotes the {\em free monoid over $A$}.
An element of $A$ is  called {\em letter} and an element of $A^*$ is  called {\em word}. The length of word $w$ is denoted by $\abs w$, and ${\abs w}_{a}$
counts how often a letter $a$ appears in $w$. 
 
If $M$ is any monoid and $u,v\in M$, then we write 
$u \leq v$ if $u$ is a {\em factor} of $v$, which means we can factorize  
$v= xuy$ for some $x,y \in M$. We denote the neutral element in $M$ by $1$. In particular, $1$ denotes also the empty word. 

An \emph{involution} of a set $A$ is a  mapping $x \mapsto \ov x$ such that 
$\overline{\overline{x}} = x$ for all $x\in A$. For example, the identity map is an \invol. 
A {\em \morph} between sets with involution is a mapping  respecting the involution. A {\em monoid with \invol} has to additionally satisfy $\overline{xy}=\overline{y}\,\overline{x}$. 
A {\em \morph} between monoids with \invol is a \hom $\phi: M \to N$ such that 
$\phi(\ov x) = \ov{\phi(x)}$. 
It is an {\em $S$-\morph} if $\phi(x) = x$ for all $x \in S \sse M$.
All groups are monoids with involution given by 
$\ov x = x^{-1}$, and all group-\homs are \morph{s}.
  Any involution on a set $A$ extends  
to 
$A^*$: for a word $w = a_1 \cdots a_m$ 
 we let  $\ov{w} = \ov{a_m} \cdots \ov{a_1}$. 
If $\ov a = a$ for all $a \in A$ then $\ov{w}$ is simply the word $w$ read from right-to-left. The monoid $A^*$ is called a {\em free monoid with involution}. 

\subsection{NFAs, rational and recognizable subsets in monoids}\label{sec:rational}
Let us recall the notion of recognizable and rational set 
and let us emphasize that this concerns incomparable families, in general.
See \cite{eil74} for more background. 
The term {\em regular} will be used only in the context of finitely generated free monoids. 

Let $M$ be any monoid. A subset $L \in M$ is called {\em recognizable} if there is 
a \hom $\psi: M \to N$ to some finite monoid $N$ such that 
$L = \oi\psi(\psi (L))$. We also say that $N$ or $\psi$ {\em recognizes} $L$.
Recognizability is a ``saturation property'': we have $w \in L$ \IFF
$\psi(w) \in \psi(L)$. 

The family of \emph{rational subsets} $\Rat(M)$ is defined inductively as follows.  All finite subsets of
$M$ are rational.  If $L,L' \subseteq M$ are rational,
then the union $L \cup L'$, the concatenation
$L \cdot L'$, and $L^+$ are rational. Define $L^0= \os 1$ and 
$L^{i+1} = L \cdot L^{i}$ for $i\in \N$. Then 
$L^+ = \bigcup\set{L^i}{i>0}$ denotes the subsemigroup of $M$ which is generated by the subset $L\sse M$. We let $L^*= L^+ \cup \os 1$; then 
$L^*$ is the corresponding submonoid. 

For a finitely generated free monoid $A^*$ the family  $\Rat(A^*)$ coincides with the family of recognizable subsets: this is the content of Kleene's classical theorem (see any standard textbook on formal languages, such as \cite{eil74}). In general, however, the two families are incomparable. For example, for a group $G$  the two families coincides \IFF $G$ is finite. 

By definition, if $h:M \to M'$ is a \hom, then 
$L\mapsto h(L)$ induces a mapping $\Rat(M) \to \Rat(M')$. The mapping 
is surjective \IFF the \hom $h$ is surjective. In the following 
let $M$ be finitely generated. Consider any surjective
\hom $\pi: \Gam^* \to M$  where $\Gam$ is finite. Then every $L\in \Rat(M)$ can be specified by some $K\in \Rat(\Gam^*)$ such that
$\pi(K) = L$. The family  
$\Rat(\Gam^*)$ coincides with the family of ``regular'' subsets of $\Gam^*$. Regular subsets of $\Gam^*$ are those which can be accepted by a non-deterministic finite automaton, i.e. an NFA. 
Again, we can define the notion of NFA for arbitrary monoids $M$: 
an NFA $\cA$ over $M$ is a tuple $\cA= (Q,M,\del,I,F)$ where
$Q$ is a set of {\em states}, $I \sse Q$ is the set of {\em initial} states, $F \sse Q$ is the set of {\em final} states,  and
$\del \sse Q\times M \times Q$ is a finite set of {\em transitions}. 
If $(p,m,q)$ is a transition then $m \in M$ is called its {\em label} and each path 
$$(p_0,m_1,p_1),(p_1,m_2,p_2)\lds (p_{k-1},m_k,p_k)$$ labels a monoid element $m_1 \cdots m_k\in M$. For states $p,q \in Q$ we  let 
$L(\cA,p,q) \sse M$ be the set of labels of paths from $p$ to $q$. 
Thus, the {\em accepted language} $L(\cA)$ of an NFA $\cA$ is
$$L(\cA) = \bigcup\set{L(\cA,p,q)}{p \in I \wedge q \in F}.$$

Moreover,  if a regular set $K\sse \Gam^*$ is accepted by some NFA $\cA$ with  state set $Q=\os{1 \lds n}$, then we can choose for a recognizing \hom $N$ the monoid of Boolean $n \times n$ matrices $\Bn$. 
Indeed, for each letter $a \in \Gam$ define a matrix $\rho(a) \in \Bn$ 
by 
\begin{equation}\label{eq:boolmat}
\begin{array}{llllll}
\rho(a)_{ij} &= \left\{\begin{array}{llllll} 
1 &  & \text{if $a \in L(\cA,i,j)$}\\
0 &&  \text{otherwise}
\end{array}\right.\end{array}
\end{equation}
Cleary, due to (\ref{eq:boolmat}) we have for all $w \in \Gam^*$ the equivalence
$w \in L(\cA) \sse \Gam^*\iff \rho(w) \in \rho(L(\cA)) \sse \Bn.$

We say that a finite monoid (with involution) $N$ has an {\em efficient} representation
if we can specify each element $m\in N$ by $\Oh(\log \abs N)$ bits and if all basic operations, like equality checking, (computing the \invol), and multiplication, are in space $\Oh(\log \abs N)$, too. For example, 
$\Bn$ has an efficient representation (if the \invol is the transposition). 

In the following, all finite monoids used to define rational constraints are assumed to have an efficient representation. 
Throughout we use the following fact: if 
$\psi_1: M \to N_1$ recognizes $L_1 \sse M$ and $\psi_2: M \to N_2$
recognizes $L_2 \sse M$, then $\psi_1 \times \psi_2: M \to N_1\times N_2$
recognizes every Boolean combination of $L_1$ and $L_2$. Thus, whenever we add a new constraint by another recognizing \hom we switch to a larger monoid. There is, however, no size explosion because we use finite monoids with an efficient representation and $\log \abs{N_1\times N_2} =\log \abs{N_1}+ \log \abs{N_2}$.

\subsection{EDT0L systems}\label{sec:edtol}
The notion of {\em EDT0L system} refers to {\em {\bf E}xtended, {\bf D}eterministic, {\bf T}able, 
{\bf 0} interaction, and {\bf L}indenmayer}. 
There is a vast literature on Lindenmayer systems, see \cite{RozS86}, with various  acronyms such as D0L, DT0L, ET0L, etc.  
The subclass EDT0L is equal to HDT0L (see \eg \cite[Thm.~2.6]{rs97vol1}), and has received particular attention. 
We use very little L-theory: essentially  we content ourselves
to define EDT0L through a characterization  (using  rational control)
due to Asveld \cite{Asveld1977}.
The class of EDT0L languages is a proper subclass of indexed languages 
in the sense of \cite{Aho68}, see 
\cite{EhrRoz77}.  For more background 
we refer 
to \cite{rs97vol1}.
\begin{definition}\label{def:edt0lasfeld}
Let $A$ be an alphabet and $L\sse A^*$ be a subset.
We say that $L$ is {\em EDT0L} if there is an alphabet $C$ with 
$A \sse C$, a finite set $H \sse \End(C^*)$ of endomorphisms of $C$, a rational language $R\sse H^*$,  
and a symbol $\#\in C$ 
such that 
$L = \set{\phi(\#)}{\phi \in R}.$ 
\end{definition}

Note that for a subset  $R \sse H^*$ of  endomorphisms  of $C^*$ we have $\set{\phi(\#)}{\phi \in R}$ is a subset of 
 $C^*$. Our definition implies that $R$ must guarantee that $\phi(\#)\in A^*$ for all 
$\phi \in R$. The language 
$C$ is called an {\em extended alphabet}.

\begin{example}\label{ex:edtvd}
Let $A= \os{a,b}$ and $C= \os{a,b, \#,\$}$. We let $H$ be set of four endomorphisms
$f,g_a,g_b,h$ satisfing $f(\#) = \$\$$, 
$g_a(\$) = \$a$, $g_b(\$) = \$b$, and $h(\$) = 1$, and on all other letters the
$f,g_a,g_b,h$ behave like the identity. 
Consider the rational language $R\sse H^*$ 
defined by 
$R = h\os{g_a, g_b}^* f$ (where endomorphisms are applied right-to-left).
A simple inspection shows  that $\set{\phi(\#)}{\phi \in R} = \set{vv}{v\in A^*}$, which is not context-free. 
\end{example}

\subsection{Triangulation}\label{sec:triang}
 We can replace an equation over any monoid $M$ by a system of
 {\em triangular equations} (\ie equations $X=V$ where $\abs V = 2$) and one additional special equation $Y=1$, where $Y$ is a fresh variable. The procedure is straightforward: 
 consider an equation $U=V$ where $U,V \in (A\cup \OO)^*$ are words, $A\sse M$, and $\OO= \os{X_{1}\ldots X_{k}}$ is a set of (free) variables. 
Clearly, 
$$\forall X_{1}\lds X_{k}: (U=V \iff \exists X,Y:\; X = UYY \wedge  X= VYY \wedge  Y=1).$$
With the exception $Y=1$, both equations have the form $X=W$ with $\abs W \geq 2$. If $\abs W \geq 2$ then $W=yzW'$ and we obtain 
$$\forall X, Y, X_{1}\lds X_{k}:X= W\iff \exists X':\; (X'=yz) \wedge  (X = X'W').$$
Using this we can transform any system of equations into a new system with additional variables; and we end up with a system of equations of  type $X=yz$ and one (singular) equation $Y=1$.

\section{Proof of Theorem~\ref{thm:freecentral}}\label{sec:pr}
\subsection{Preprocessing.}\label{sec:preproc}
According to \prref{sec:triang} we start with a system of triangular  
of equations: there is one special equation $Y=1$ and all other equations have the form $X=V$ where $X$ is a variable and $\abs V = 2$.
Next, we add the special symbol $\#$ to the alphabet $\Apone$; and we define $A = \Apone \cup \os \#$. We let $\ov \# = \#$; this will be the only self-involuting letter in the proof of \prref{thm:freecentral}. 
We must make sure that no solution uses $\#$ and every solution is in reduced words. 
We do so by introducing a finite monoid $N$ with involution which 
plays the role of (a specific)
rational constraint. 

Let $N = \os{1,0} \cup \Apone \times \Apone$ have multiplication $1\cdot x = x \cdot 1 = x$, 
$0\cdot x = x \cdot 0 = 0$, and 
\begin{equation}\label{eq:nfm}
\begin{array}{llllll}
(a,b)\cdot (c,d) = \left\{\begin{array}{llllll} 
0 & \mathrm{if} & b=\ov c\\
(a,d) &&  b\neq \ov c.
\end{array}\right.\end{array}
\end{equation}
The monoid $N$ has an involution by $\ov 1 = 1$, $\ov 0 = 0$, and 
$\ov{(a,b)} = (\ov b , \ov a)$. Consider the \morph
$\mu_{0}:A^* \to N$ given by $\mu_{0}(\#) =0$ and $\mu_{0}(a) = (a, a)$ for $a 
\in \Apone$. It is clear that $\mu_{0}$ respects the involution and $\mu_{0}(w) = 0$ if and only if
either $w$ contains $\#$ or $w$ is not reduced.
If, on the other hand, $1 \neq w\in \Apone^*$ is reduced, then 
$\mu_{0}(w) =(a,b)$, where $a$ is the first and $b$ the last letter of $w$. Also, $\mu_{0}(w) =1$ \IFF $w = 1$.
Thus, if $\sig$ is a solution in reduced words, then for each 
variable $X$ there exists some $0 \neq \mu(X)= \ov {\mu(\ov X)}\in N$ with 
$\mu(X)= \mu_{0}\sig(X)$. 

The \morph $\mu$ allows us to remove the special equation $Y=1$. It is replaced by fixing $\mu(Y)=1$.

For the other variables there are only finitely many choices for $\mu$, so we assume that each 
equation is specified together with a morphism $\mu:\OO\ra N$.
We now require that a solution $\sig:\OO \to A^*$ satisfy 
three properties: $\pi\sig(U) = \pi\sig(V)$, $\ov{\sig(X)}= \sig(\ov X)$, and $\mu(X)= \mu_{0}\sig(X)$ for all $X\in \OO$.

The next step allows us to work with free monoids with involution rather than with groups. This relies on \prref{lem:caytree},
which is well-known, too. Its geometric interpretation is simply that the Cayley graph of a free group (over standard generators) is a tree. 
See \prref{fig:xyz} for a visual explanation. 
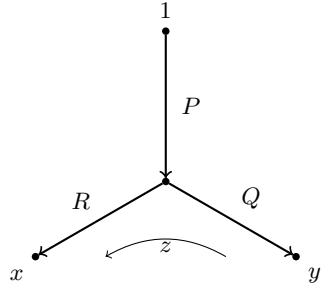
\begin{figure}[h]
\begin{center}
\begin{tikzpicture}[outer sep=0pt, inner sep = 1pt, node distance = 8pt]
 \node[circle, fill] (x) at (0,0) {};   
 \node[circle, fill] (a) at (90:2) {};
 \node[circle, fill] (b) at (210:2) {};
 \node[circle, fill] (c) at (330:2) {};
 
 \begin{scope}[very thick,decoration={
    markings,
    mark=at position 1 with {\arrow{>}}}] 
 \draw[thick, postaction={decorate}] (a) -- node[left= 5]  {} node[right =5] {$P$} (x); 
   \draw[thick, postaction={decorate}] (x) -- node[below right= 3]  {} node[above left =3] {$R$} (b);
 \draw[thick, postaction={decorate}] (x) -- node[above right = 3]  {$Q$} node[below left=3] {} (c);
 \end{scope}
 
 \node[below left of  = b, node distance = 10pt] {$x$};
 \node[below right of  = c, node distance = 10pt] {$y$};
 \node[above of  = a] {$1$};
 \path (0.8,-1) edge [bend right, ->] node [below] {$z$}   (-0.8,-1);
\end{tikzpicture}
\begin{align}
\forall x,y,z \; \exists P,Q,R : \;
\pi(x) = \pi(yz)
 \iff  
 \pi(x) = PR \wedge  \pi(y) = PQ \wedge \pi(z) = \ov Q R.
\end{align}
\caption{Part of the Cayley graph of $\FG \Apos$ with standard generators: the root is $1$ on the top. The geodesics to vertices $x$ and $y$ split after an initial path labeled by  $P$.}
\label{fig:xyz}\end{center}
\end{figure}\qed

\begin{lemma}\label{lem:caytree}
Let $x,y,z$ be reduced words in $\Apone^*$. Then $x=yz$ in the group $\FG \Apos$ 
\IFF there are reduced words $P,Q,R$ in $\Apone^*$ such that
$x= PR$, $y = P Q$, and $z =  \ov Q R$ holds in the free monoid $\Apone^*$. 
 \end{lemma}

\begin{proof}
If there are words $P,Q,R$ in $\Apone^*$ such that
$x= PR$, $y = P Q$, and $z =  \ov Q R$ holds in the free monoid $\Apone^*$ then we have $\pi(x) = \pi(yz)$, hence $x=yz$ in the group $\FG \Apos$. This is trivial and holds whether or not $x,y,z$ are reduced. For the other direction, let $x,y,z$ be reduced words in $\Apone^*$ such that $\pi(x) = \pi(yz)$. If the word $yz$ is reduced
then we choose $P= y$,  $Q= 1$, and $R=z$. In the other case
we have $y= y'a$ and  $z= \ov a z'$ for some $a \in \Apone$ and we can use induction. \qed
\end{proof}

The  consequence of \prref{lem:caytree} is that  with the help of fresh variables $P,Q,R$  we can substitute every equation 
$x= yz$ with $x,y,z\in \Apone \cup \OO$ 
by the following three word equations to be solved over a free monoid with involution.
\begin{align}\label{eq:fg2fm}
x = PR,\qquad  y = PQ,\qquad z = \ov Q R.
\end{align}
The enlarged set of variables becomes $\OO \cup \os{P,\ov P,Q,\ov Q,R,\ov R}$.
Assume $\sig$ solves (\ref{eq:fg2fm}) such that $\sig(x)$, $\sig(y)$, $\sig(z)$ are reduced. Let $\sig(P) =p$, $\sig(Q) =q$, and $\sig(R) =r$,
Then  $p$, $q$, and $r$ are reduced and there is no cancellation 
in any of the words $pr$, $pq$, and $\ov q r$. The lack of cancellation is encoded by $\mu$, but note that there are various choices for $\mu(P), \mu(Q), \mu(R)$: for example, let $\mu(x) = (b,c)$, $\mu(y) = (b,a)$, and $\mu(z) = (\ov a,c)$ and assume that
$p,q,r \neq 1$. Then there are a last letter $d$ in $p$ and first letters 
$e$ in $q$ and $f$ in $r$ such that $\ov d \neq e \neq f \neq \ov d$,  
so the choice $\mu(P) = (b,d)$, $\mu(Q) = (e,a)$, and 
$\mu(R) = (f,c)$ is correct.

\subsection{The initial equation}\label{sec:isup}
By \prref{lem:caytree} it is enough to prove the analogue of \prref{thm:freecentral} for free monoids with involution, systems of equations $(U_{i},V_{i})$, $1 \leq i \leq s$ and a \morph $\muinit:(A\cup \OO)^* \to N$ such that   
$0 \neq \muinit(U_{i}) = \muinit (V_{i})$ for all $i$. Finally, we construct a single equation $(U', V')$ over $A \cup \OO$, where $U'= U_{1}\# \cdots \# U_{s}$ and 
$V'= V_{1}\# \cdots \# V_{s}$. 
Notice that $\muinit(X) \neq 0$ and $\abs{U_{i}}_{\#} = \abs{V_{i}}_{\#}= 0$ for all $i$.
 A solution 
$\sig:\OO \to \Apone^*$ must satisfy $\sig(U')=\sig(V')$, $\ov {\sig(X)} = \sig(\ov X)$, and $\muinit(X)= \mu_{0}\sig(X)$ for all $X\in \OO$. 
 The set of variables $\os{X_{1}\lds X_{k}}$, specified in \prref{thm:freecentral}, is a subset of the new and larger set of variables $\OO$, and the original solution set became a finite union
 of solution sets with respect to different choices for $\mu_{0}$ and 
 $\muinit$. The result holds since EDT0L is closed under finite union. 
In order to achieve our result we have to \emph{protect} each variable $X_{i}$ as follows. We assume  $\Apone \cup \OO = \os{x_{1}\lds x_{\ell}}$ with $x_{i}= X_{i}$ for $1 \leq i \leq k$ where $\{X_1\dots, X_k\}$ is the specified subset in the statement of \prref{thm:freecentral}.

The word $\Winit$ over $(A \cup \OO)^*$ is then defined as: 
\begin{equation}\label{eq:Winit}
\Winit=  \#x_1\#\dots \#x_\ell \# U'\# V'\#\ov {U'}\# \ov {V'}\# \ov{x_\ell}\#\dots \#\ov{x_1}\#.
\end{equation}
 Observe that $\Winit$ is longer than (but still linear in) $\abs A + \abs \OO + \abs{UV}$.  
The number of $\#$'s in $\Winit$ is odd; and if $\sig:(A\cup\OO)^*\to A^*$ is a morphism with $\sig(X)$ reduced for all $X\in \OO$, then: 
 $
\pi\sig(U) = \pi\sig(V) \iff \sig(U') = \sig(V') \iff \sig(\Winit) = \sig(\ov{\Winit}).
$ 
Here $(U,V)$ is the equation in \prref{thm:freecentral} and $(U',V')$
is the intermediate word equation over $A \cup \OO$. Therefore 
\prref{thm:freecentral} follows by showing that the following language is EDT0L:
$$\set{\sig(X_1)\# \cdots \#\sig(X_k)}{
\sig(\Winit) = \sig(\ov{\Winit}) \wedge \muinit= \mu_{0}\sig \wedge 
\forall X: \sig(\ov{X}) =\ov{\sig(X)}}.$$

\subsection{Partial commutation and extended equations.}\label{sec:exeqn}

Partial commutation is an important concept in our proof.
It pops up where traditionally  the unary case (solving a linear \Dio equation) is used as a black box, as is done in \cite{DiekertJP2014csr}. At first glance it might seem  like an unneccesary complication, but  in fact the contrary holds. Using partial commutation allows us to encode all solutions completely in the edges of a graph, which we can construct in quasi-linear space,
and is 
one of the major differences to \cite{DiekertJP2014csr}.
As a (less important) side effect, results on linear \Dio equations
come for free as this is the special case $\FG \Apos = \Z$: 
solving linear \Dio equations becomes part of a more general process. 

We fix $n = \ninit= \abs {\Winit}$ 
and some $\kappa \in \N$ large enough, say $k=100$. We let $C$ 
be an alphabet with involution (of constants)  such that $\abs C =\kappa n$ and $A \sse C$. 
We define $\Sig= C\cup \OO$ and assume that $\#$ is the only self-involuting symbol of $\Sig$. In the following $x,y,z, \ldots$ refer to words in 
$\Sig^*$ and $X,Y,Z, \ldots$ to variables in $\OO$. 
Throughout we let $B,B'$ and $\cX,\cX'$ denote 
subsets which are closed under  involution and satisfy
 $\cX' \sse \cX \sse \OO$ and either $A \sse B\sse B' \sse C$ or
 $A \sse B'\sse B\sse C$. In particular, $B$ and $B'$ are always comparable.

 We encode a partial commutation as follows. 
 Let $c \in B$ and either $p=c$ or $p=c\ov c$. (The case $p = c \ov c$ is needed only later in the proof of \prref{thm:procentral}.)
Let $\theta\sse (\cX \cup B^+)\times \os{p,\ov p}$ denote an irreflexive and antisymmetric relation. It is called a {\em type} if $\theta$  satisfies the following conditions:
 first,  
$(x,y)\in \theta$ implies $(\ov x,\ov y)\in \theta$ and  
$x \in \cX \cup \set{a,\ov a, a \ov a}{a \in B}\sm \os{c,\ov c, c \ov c}$, and second, for each $x$ we have $\abs {\theta(x)} \leq 1$, where $\theta(x) = \set{y \in B^*}{(x,y)\in \theta}$. 
The type relation  $\theta$ can be stored in \qls.
Given $\theta$ and $\mu: B\cup \cX \to N$ such that 
$\mu(xy) = \mu(yx)$ for all $(x,y) \in \theta$, we define the following two partially commutative monoids with \invol:  
first,  $M(B,\cX,\theta,\mu) = (B\cup \cX)^*/\set{xy = yx}{(x,y) \in \theta}$ -- a monoid with a \morph $\mu: M(B,\cX,\theta,\mu) \to N$ -- and second,  $M(B,\theta,\mu) = B^*/\set{xy = yx}{(x,y) \in \theta}$. If $w\leq W\in M(B,\cX,\theta,\mu)$ then $w$ is called a {\em proper
factor} if $w \neq 1$ and $\abs{w}_\# = 0$. 

Since the defining relations for $M(B,\cX,\theta, \mu)$ are of the form $xy=xy$,
we can define $\abs W$ and $\abs{W}_a$ for $W\in M(B,\cX,\theta, \mu)$ by representing $W$ by some word $W\in (B\cup \cX)^*$. 
 Typically we represent $w$, $W$ by words $w,W\in (B\cup \cX)^*$,
but their interpretation is always in $M(B,\cX,\theta, \mu)$.

\begin{definition}\label{def:wellf2}
We call $W\in M(B,\cX, \theta, \mu)$ {\em well-formed} if  $\abs W \leq \kappa n$, ${\abs W}_{\#} = {\abs \Winit}_{\#}$,    
every proper factor $x$ of $W$ and every $x\in B \cup \cX$ satisfies $\mu(x) \neq 0$,  and $B^* \cap \oi \mu(1) = \os 1$.
 Moreover, if $x$ is a proper factor then $\ov x$ is a proper factor, too. Finally, for every $a \in \Apone$ there is a factor $\#a\#$ in $W$.
\end{definition}

\begin{definition}\label{def:extequat}
 An {\em extended equation} is a tuple $V= (W,B,\cX,\theta,\mu)$ 
where $W \in M(B,\cX, \theta, \mu)$ is well-formed. 

A {\em $B$-\solu} of $V$ is a $B$-\morph $\sig:M(B,\cX,\theta,\mu)\to M(B,\theta,\mu)$ such that 
$\sig(W) = \sig(\ov W)$ and $\sig(X) \in y^*$ whenever $(X,y)\in \theta$.

A {\em \solu} of $V$ is a pair  $(\alp,\sig)$  such that 
$\alp: M(B,\theta,\mu) \to  A^*$ is an $A$-\morph satisfying $\mu_{0}\alp = \mu$ and $\sig$ is a $B$-\solu.  
\end{definition}

\begin{mdframed}
\begin{tabular}{ll}
$W$ & = equation, where the solution is a ``palindrome'' $\sig(W) = \ov {\sig(W)}\in A^*$. \\ 
$B$ & = alphabet of constants with $\#\in A\sse B= \ov B \sse C$. \\ 
$\cX$ & = variables appearing in $W$. Hence, $\cX = \ov \cX \sse \OO$. \\    
$\mu$ & = \morph to control the constraint that the \solu is reduced.\\
$\theta$ & = partial commutation.
\end{tabular}
\end{mdframed}
During the process of finding a solution, we change these parameters, and we describe the process in terms of  a diagram (directed graph) 
of states and arcs between them.

\subsection{The directed labeled graph $\cG$.}\label{subsec:myG}
We are now ready to define the directed labeled graph $\cG$ which will be the core of the NFA defining the EDT0L language 
$\cSol(U,V)$ 
we are aiming for.

Define the {\em vertex set} for $\cG$ to be the set of all extended equations 
$V= (W,B,\cX,\theta,\mu)$. 
The {\em initial vertices} are of the 
form $(\Winit,A,\OO,\es,\muinit)$. Due to the different possible choices for 
$\muinit$ there are exponentially many
initial vertices.  In fact, \prref{thm:freecentral} requires us to work in exponentially many graphs simultaneously because of the different possible
$\mu_{0}: A^* \to N$. However, to simplify the presentation we fix one $\mu_{0}$. 

We define the set of {\em final vertices} by 
$
\set{(W,B,\es,\es,\mu)}{W= \ov W}.
$ 
 By definition every final vertex trivially has a $B$-solution $\sig= \id B$. (Note that in a final vertex there are no  variables.)  
The arcs in $\cG$ are labeled and are of the form 
$(W,B,\cX,\theta,\mu) \arc h (W',B',\cX',\theta',\mu')$. Here
$h: C^* \to C^*$ is an \Endo which is given by 
a \morph $h: B' \to B^*$ such that 
$h$ induces a well-defined \morph
$h:M(B' \cup \cX',\theta',\mu') \to  M(B \cup \cX,\theta,\mu)$. Note that the direction of the \morph is opposite to the direction of the arc. 
There are further restrictions on arcs. For example, 
we will have  $\abs{h(b')}\leq 3$ for all $b'$. The main idea is as follows.
Suppose $(W,B,\cX,\theta,\mu) \arc h (W',B',\cX',\theta',\mu')$ is an arc, 
$\alp: M(B,\theta,\mu)\to M(A,\es,\mu_{0})$ is an $A$-\morph, and
$(W',B',\cX',\theta',\mu')$ has a $B'$-solution $\sig'$; then 
there exists a solution $(\alp, \sig)$ of the vertex $(W,B,\cX,\theta,\mu)$.
Moreover, for the other direction if $(\alp, \sig)$ solves $V= (W,B,\cX,\theta,\mu)$ and $V$ is not final then we can follow an outgoing arc and recover $(\alp, \sig)$ from a solution at the target node. We will make this more precise below.

\subsection{Compression arcs.}
These arcs transform the sets of constants.
Let  $V= (W,B,\cX,\theta,\mu)$ and $V'= (W',B',\cX,\theta',\mu')$ be vertices in $\cG$. The  compression arcs have the form $V\arc h V'$, where 
 either  $h=\eps$ is defined by the identity on $C^*$, or $h$ is defined by 
a mapping $c \mapsto h(c)\neq c$ with $c \in B'$. Recall that if a \morph $h$ is defined by $h(c) = u$ for some letter $c$ then, by our convention, $h(\ov c) = \ov u$ and $h(x) = x$ for all $x \in \Sig$ which are different from $c$ and $\ov c$.  We assume  $0 \neq \mu'(c) = \mu(h(c))\neq 1$ and 
$\mu(x) = \mu'(x)$ for all $x\in (B\cap B') \cup \cX$ (if not explicitly stated otherwise). 

We define compression arcs $(h(W'),B,\cX,\theta,\mu) \arc h (W',B',\cX,\theta',\mu')$  of the following types, only. 

\begin{enumerate}
\item[\df 1.]   ({\bf Renaming.})
We assume that $h$ is defined by $h(c) = a$ such that $B \varsubsetneq B'=B\cup\{c,\ov c\}$, and 
 $\theta \sse \theta'$.
 Thus, possibly, $\theta \varsubsetneq \theta'$. 

\item[\df 2.]   ({\bf Compression.}) We have $h(c) = u$ with $\abs u \leq 3$
and either 
 $B=B'$ and  $\theta'= \theta$ or $B \varsubsetneq 
 B'=B\cup\{c,\ov c\}$ and $\theta = \theta'= \es$.

 \item[\df 3.]   ({\bf Alphabet reduction.}) We have
 $
 B'\varsubsetneq  B$, 
 $\theta'= \es $, and $h$ is induced by the inclusion 
 $B' \sse B$. Hence we have $h = \eps = \id {C^*}$.

\end{enumerate} 
For the proof of \prref{thm:freecentral} we compress words $u$ into a single letter $c$
only if $\abs u \leq 2$. For \prref{thm:procentral} we will have in addition
the case where $u= a \ov a c$ with either $a= c$ (and $\ov a=\ov c$) or $(a \ov a , c \ov c) \in \theta$.
The purpose of arcs of type  $ \df 3$ is to remove letters in $B$ that do not appear in the word $W$. This allows us to reduce the size of $B$ and to ``kill'' partial commutation.  
  
\begin{lemma}\label{lem:contr}
Let $(W,B,\cX,\theta,\mu)\arc h (W',B',\cX',\theta',\mu')$  be an arc of type \df{1,2} or \df{3} with $W= h(W')$.
Let $\alp: M(B,\theta,\mu) \to  M(A,\emptyset,\mu_{0})$ be an $A$-\morph 
at vertex $V= (h(W'),B,\cX,\theta,\mu)$ and 
$\sig'$ be a $B'$-solution to $V'=(W',B',\cX',\theta',\mu')$.
Define a $B$-\morph $\sig: M(B,\cX,\theta,\mu) \to  M(B,\theta,\mu)$
by $\sig(X)=h\sig'(X)$. Then 
$(\alp,\sig)$ is a \solu at $V$ and $(\alp h,\sig')$ is a \solu at $V'$
and 
$\alp\sig(W) = \alp h\sig'(W')$. 
\end{lemma}

\begin{proof} 
By definition, $\mu h = \mu'$ and $\mu_{0}\alp= \mu$. Hence $(\alp h,\sig')$ is a \solu at $V$. Now,  
$h(X)= X$ for all $X\in \cX$. 
Hence, $\sig(h(X))=\sig(X)=h\sig'(X)$. For $b'\in B'$ we obtain $\sig h(b')=h(b')=h \sig'(b')$ since $\sig'$ and $\sig$ are the identity on $B'$ and $B$ respectively. It follows $\sig h=h\sig'$ and $\alp\sig(W) = \alp h\sig'(W')$. Next, 
$$
\sig(W) = \sig(h(W')) = h (\sig'(W')) = h (\sig'(\ov{W'}))= \sig (h (\ov{W'})) 
= \sig(\ov{h(W')})= \sig(\ov{W}).
$$ 
Thus, $\sig$ is a $B$-\solu to $V$ and, consequently, $(\alp,\sig)$ solves $V$.
\qed
\end{proof}

\subsection{Substitution arcs.}
These arcs transform variables.
Let $V= (W,B,\cX,\theta,\mu)$ and  
$V'= (W',B,\cX',\theta',\mu')$ be vertices in $\cG$ and $X\in \cX$.
We assume that $\cX = \cX' \cup \os{X, \ov X}$ and 
$\mu(x)= \mu'(x)$, as well as 
$\theta(x)= \theta'(x)$ for all $x \in (B\cup \cX) \sm \os{X, \ov X}$.
The set of constants is the same on both sides, but $\cX'$ might have fewer variables.  
Substitution arcs are defined by a \morph $\tau:\os X \to  BX \cup \os 1$ such that we obtain a $B$-\morph $\tau:M(B,\cX,\theta,\mu)\to M(B,\cX',\theta',\mu')$.
We let $\eps=\id{C^*}$ as before.
 We define substitution arcs $(W,B,\cX,\theta,\mu) \arc \eps (\tau(W),B,\cX',\theta',\mu')$ 
  if one of the following conditions apply.

\begin{enumerate}
\item[\df 4.]   ({\bf Removing a variable.}) 
Let $\cX' = \cX \sm \os{X, \ov X}$. The  $B$-\morph
$\tau:M(B,\cX,\theta,\mu)\to M(B,\cX',\theta',\mu')$ is defined by $\tau(X)=1$.    

\item[\df 5.]   ({\bf Variable typing.}) 
The purpose of this arc is to introduce some type for variables without changing anything else, so $\cX'=\cX$ and $\mu'=\mu$. Suppose that $\theta(X) = \es$ and $p \in B^+$ is a word with  
 $\mu(Xp) = \mu(pX)$ and such that 
$\theta'=\theta\cup\os{(X,p),(\ov X, \ov p)}$. The  $B$-\morph $\tau:M(B,\cX, \theta,\mu)\to  M(B,\cX, \theta',\mu)$ is defined by the identity on $B\cup\cX$. 
Note that the condition $\mu(Xp) = \mu(pX)$ implies that if $\mu:M(B,\cX,\theta,\mu) \to N$ is well-defined, then $\mu:M(B,\cX,\theta',\mu) \to N$ is well-defined, too. The other direction is trivial. 

\item[\df 6.] ({\bf Substitution of a variable.}) 
We have $(B,\cX,\theta) = (B',\cX',\theta')$. 
Let $p\in B^+$ such that $\theta(X) \sse \os p$. (For $\theta(X)= \es$ this is always true.) We suppose that we have $\mu(X) = \mu(p) \mu'(X)$ (hence, automatically $\mu(\ov X) = \mu'(\ov X)\mu(\ov p)$) and that $\tau(X)  =  pX$ defines a morphism $\tau: M(B,\cX, \theta,\mu)\to M(B,\cX, \theta,\mu')$. 
\end{enumerate}

\begin{lemma}\label{lem:varsub}
Let $V= (W,B,\cX,\theta,\mu)\arc \eps (W',B,\cX',\theta',\mu')= V'$ with $\eps = \id{C^*}$ be an arc of type \df{4,5,6} with $W'= \tau(W)$.
Let $\alp: M(B,\theta,\mu) \to  M(A,\emptyset,\mu_{0})$ be an $A$-\morph 
at vertex $V$
and 
$\sig'$ be a $B$-solution to $V'$. 
Define a $B$-\morph $\sig: M(B,\cX,\theta,\mu) \to  M(B,\theta,\mu)$
by $\sig(X)=\sig'\tau(X)$. Then 
$(\alp,\sig)$ is a \solu at $V$ and $(\alp,\sig')$ is a \solu at $V'$.
Moreover, $\alp\sig(W) = \alp h\sig'(W')$ where $h=\eps$ is viewed 
as the identity on $\id{M(B,\theta,\mu)}$.
\end{lemma}

\begin{proof}
 Since $\sig'$ is a $B$-solution to $V'$ we have 
 $\sig(W)= \sig'(\tau(W))= \sig'(\ov {\tau(W)})= \ov{\sig'\tau(W)}= \ov{\sig(W)}.$
 Hence, $(\alp,\sig)$ is a \solu at $V$. Since $M(B,\theta,\mu) = M(B,\theta',\mu')$ 
 (a possible change in $\mu$ or $\theta$ concerns variables, only),
 $(\alp,\sig')= $ is a \solu at $V'$. The assertion 
  $\alp\sig(W) = \alp h\sig'(W')$ is trivial since $W'= \tau(W)$, 
  $\sig=\sig'\tau$, and $h=\eps$ induces the identity on $M(B,\theta,\mu)$. 
  \qed
\end{proof}

\begin{proposition}\label{prop:backandforth}
Let $V_{0}\arc{h_1} V_1 
\cdots \arc{h_{t}} V_t$
be a path in $\cG$ of length $t$, where $V_0= (\Winit, A, \OO,\es ,\muinit)$
is an initial and $V_t=(W',B,\es,\es,\mu)$ is a final vertex. 
Then $V_{0}$ has a solution $(\id{A},\sig)$ with
 $\sig(\Winit)= h_1 \cdots h_t(W')$. 
 Moreover, 
 we have 
 $W'\in  \#u_{1}\#\cdots \#u_{k}\#B^*$ such that $\abs{u_{i}}_{\#}= 0$ and 
 we can write:
  \begin{align}\label{eq:wondercX} 
h_1 \cdots h_t(u_{1}\#\cdots \#u_{k}) = \sig(X_{1})\#\cdots \#\sig(X_{k}),
\end{align}
\end{proposition}

\begin{proof}
By definition of final vertices we have $\ov{W'} = {W'}$ and no variables occur in $W'$. Hence, $\id {B^*}$ defines the (unique) $B$-\solu of $W'$. 
By definition of the arcs, $h=h_1 \cdots h_t: M(B,\es, \mu) \to A^*= M(A,\es, \muinit) $  is an $A$-\morph which shows that
$(h,\id {B^*})$ solves $W'$. There is only one $A$-\morph 
at $V_{0}$, namely $\id {A^*}$. 
Using \prref{lem:contr}
 and \prref{lem:varsub} we see first, $V_{0}$ has some solution 
 $(\id {A^*},\sig)$ and second, 
\begin{align}\label{eq:wonder} 
\id {A^*}\sig(\Winit) = \id {A^*}h_1 \cdots h_t\id {B^*}(W') = 
h_1 \cdots h_t(W').
\end{align}
Finally, for $1 \leq j \leq t$ we have $h_{j}(\#) = \#$ and $\abs{h_{j}(x)}_{\#} = 0$
for all other symbols. Hence the claim 
$h_1 \cdots h_t(u_{1}\#\cdots \#u_{k}) = \sig(X_{1})\#\cdots \#\sig(X_{k})$. 
\qed
\end{proof}

\subsection{Construction of the NFA $\cA$ in quasilinear space}\label{sec:cinqls}
The input to \prref{thm:freecentral} is given by three items: 
$\Apone$, $\OO$, and the pair of words $U,V\in (\Apone \cup \OO)^*$. 
The input size in bits is in
$$\Oh((\abs{UV}+\abs{\Apone \cup \OO}) \log(\abs{\Apone \cup \OO})).$$
Our non-deterministic procedure will use space 
$\Oh((\abs{UV}+\abs{\Apone \cup \OO}) \log(\abs{UV}+\abs{\Apone \cup \OO}))$,
which is \ql in the input size.

Let $n_{0}= (\abs{UV}+\abs{\Apone \cup \OO}) \log(\abs{\Apone \cup \OO}))$ be the input size. 
The definition of $n=\ninit = \abs{\Winit}$ gives
$n \in \Oh(n_0)$, and therefore we can also choose $n= \ninit$ as the new input size. 
In the preprocessing we transformed the original input into 
the word $\Winit$ over a larger set of variables. However, transforming the initial data into $\Winit$ is easily implementable in \qls. 

The next observation is that, according to \prref{def:wellf2}, every well-formed word $W$ 
can be stored with $\Oh(n \log n)$ bits. Moreover, given 
$W\in M(B,\cX,\theta,\mu)$, it can be checked (deterministically) in space $\Oh(n \log n)$ whether it is well-formed. 
The same is true for extended equations according to 
\prref{def:extequat}. Next, we can consider in space $\Oh(n \log n)$, one after another, all candidates 
$(W,B,\cX,\theta,\mu) \arc h (W',B',\cX,\theta',\mu')$ for arcs in $\cG$. 
(We use that, by construction,  $h$ can be encoded by a list of length 
$\Oh(n)$ for all arcs in $\cG$.) 
Each time we check (deterministically) in space $\Oh(n \log n)$ whether there is indeed an arc $h$. If the answer is positive, we output the arc. 
The switch from the graph $\cG$ to the NFA $\cA$ is computationally easy: 
for every final vertex $(W',B,\es,\es,\mu)$ compute 
a prefix $\#w'\#$ of $W$ with $\abs{w'}_{\#} = k$ and output the
arc $(W',B,\es,\es,\mu)\arc {g_{w'}} \#$.

We used nondeterminism during the procedure, but this was not really necessary, as we could have produced an NFA satisfying 
\prref{thm:freecentral} in deterministic \qls. However, such an NFA would  
contain a lot of unnecessary states and arcs. 

In order to justify 
the last sentence in the abstract that deciding the existential theory of non-abelian free groups is in $\NSPACE(n\log n)$, (i.e., nondeterministic \qls) we need to check whether $L(\cA) \neq \es$.
Thus we integrate this check into the construction of the NFA $\cA$ right-away. So we modify our construction: 
instead of computing all vertices and all arcs 
of $\cG$ we output vertices  and arcs only if they belong to a path from an initial to 
a final vertex. 
More precisely, for each vertex $V$ we define a Boolean variable $\text{Useful}(V)$ 
which is $1$ if $V$ is on a path from an initial to 
a final vertex and $0$ otherwise. 
In order to compute $\text{Useful}(V)$ we guess such a path in 
nondeterministic space $\Oh(n\log n)$.  Here we use the standard fact that the nondeterministic complexity class $\NSPACE(n \log n)$ is closed under complementation by Immerman–Szelepcs{\'e}nyi (1987), see \cite{pap94}. 
Thus, our procedure outputs a vertex $V$ (resp.~an arc $V=(W,B,\cX,\theta,\mu) \arc h (W',B',\cX,\theta',\mu')= {V'}$) only if $\text{Useful}(V) =1$ (resp.~$\text{Useful}(V) = \text{Useful}(V')=1$). If no vertex satisfies $\text{Useful}(V) =1$, then we can output a one-state NFA without final states because then $\Winit$ has no solution. 
If however, at least one vertex $V$ satisfies $\text{Useful}(V) =1$, then the output is an NFA $\cA$ which accepts a nonempty set of \Endos over $C$ and, according to our construction, $\Winit$ has at least one solution; moreover, 
the reader is invited to show that $\Winit$ has infinitely many solutions \IFF  $L(\cA)$ is infinite.

\subsection{Forward property of arcs}\label{sec:fpoarcs}

The previous section has shown how to produce the graph $\cG$, which is now at our disposal. 
For every initial vertex $V_{0}$ with a given 
solution $(\id{A^*},\sig_{0})$, we 
need to establish the  existence 
of a path $V_{0}\arc{h_1} V_1
\cdots \arc{h_{t}} V_t$ to some final vertex 
$V_t=(W',B,\es,\es,\mu)$ such that the following equation holds
\begin{align}\label{eq:forinv} 
\sig_{0}(W) = h_{1}\cdots h_{t}(W').
\end{align}
This relies on the following technical concept. 
\begin{definition}\label{def:fp}
Let $V=(W,B,\cX,\theta,\mu) \arc h (W',B',\cX',\theta',\mu')= V'$ be an arc in $\cG$ and $(\alp,\sig)$ be a \solu at $V$. 
We say that the tuple 
$(V\arc h V', \alp, \sig)$ satisfies the {\em forward property} 
if there exists a \solu $(\alp h,\sig')$ at $V'$ such that 
$$ \alp \sig(W) = \alp h \sig'(W').$$
\end{definition}

\begin{lemma}\label{lem:ftaus}
Let $V= (W,B,\cX,\theta,\mu) \arc \eps (\tau(W),B,\cX',\theta',\mu')=V'$ be 
a \subst arc as in $\df 4$ or $\df 6$  and  
$(\alp,\sig)$ be a \solu at $V$. Suppose that 
$\sig(X) = uv$ and $\tau(X) = uX$. Then 
$(V\arc h V', \alp, \sig)$ satisfies the forward property. 
\end{lemma}

\begin{proof}
 If we let $\sig'(X) =v$ and $\mu'(X) = \mu(v)$ then we can write 
 $\sig= \tau \sig'$ and the \morph $\sig: M(B,\cX,\theta,\mu)
 \to M(B,\theta,\mu)$ factorizes through $\sig': M(B,\cX',\theta',\mu')
 \to M(B,\theta,\mu)$. \qed
\end{proof}

\begin{lemma}\label{lem:ftauarcs}
Let $V= (W,B,\cX,\theta,\mu) \arc \eps (\tau(W),B,\cX,\theta',\mu)=V'$ be 
a variable-typing arc as in $\df 5$ and  
$(\alp,\sig)$ be a \solu at $V$ such that $\sig(X) \in c^*$. 
(Thus,  we have $\mu(Xc) = \mu(cX) \in \mu(c^*)$; and the arc with 
the new type $(X,c)\in \theta'$ is defined.) Then 
$(V\arc h V', \alp, \sig)$ satisfies the forward property. 
\end{lemma}

\begin{proof}
 The \morph $\sig: M(B,\cX,\theta,\mu)
 \to M(B,\theta,\mu)$ factorizes canonically through $\sig': M(B,\cX,\theta',\mu)
 \to M(B,\theta,\mu)$. \qed
\end{proof}

\begin{lemma}\label{lem:compiarcs}
Let $V= (h(W),B,\cX,\theta,\mu) \arc h (W',B',\cX,\theta',\mu')=V'$ be 
any compression arc as in $\df 1,\df 2$ or $\df 3$ and  
$(\alp,\sig)$ be a \solu at $V$.
Suppose there exists a $B'$-\solu $\sig'$ at $V'$ such that 
$\sig: \cX \to M(B,\theta,\mu)$ factorizes through \morph{s} as follows
$$\sig: \cX \arc {\sig'} M(B',\theta',\mu') \arc h M(B,\theta,\mu).$$
 Then $(\alp h,\sig')$ is a \solu at $V'$ with $ \alp \sig(W) = \alp h \sig'(W')$. In particular, $(V\arc h V', \alp, \sig)$ satisfies the forward property. 
\end{lemma}

\begin{proof}
 Trivial. \qed
\end{proof}

It is clear that \prref{lem:compiarcs} does not suffice for our purpose. 
We cannot prevent $\sig$ from using letters from $B$ which are not present 
in $B'$, but then no factorization 
$\sig: \cX \arc {\sig'} M(B',\theta',\mu') \arc h M(B,\theta,\mu)$ exists if $h$ is induced by the 
identity, as in the case of alphabet reduction. As already mentioned in the 
main body of the text, we need this type of 
alphabet reduction only over empty type relations. 
We content ourselves with the following statement. 

\begin{lemma}\label{lem:alpred}
Let $V= (W,B,\cX,\es,\mu) \arc \eps (W',B',\cX,\es,\mu')=V'$ be an
alphabet reduction as in $\df 3$, 
where $B'\varsubsetneq  B$
and $\mu'$ is the restriction of $\mu$. Let $(\alp,\sig)$ be a \solu at $V$.
Define a $B'$-\morph $\bet: M(B,\es,\mu) \to M(B',\es,\mu')$ 
by $\bet(b)= \alp(b)$ for $b \in B\sm B'$ and define 
$\sig'(X) = \bet \sig (X)$.
 Then $(\alp h,\sig')$ is a \solu at $V'$ with $ \alp \sig(W) = \alp \eps \sig'(W')$. In particular, $(V\arc h V', \alp, \sig)$ satisfies the forward property. 
\end{lemma}

\begin{proof}
Since $\alp:  M(B,\es,\mu) \to M(A,\es,\mu_{0})$ is a \morph, we 
have $\mu\bet(b)= \mu \alp(b) = \mu_{0}\alp(b) = \mu(b)$ for all $b\in B\sm B'$ and $\bet$ is indeed a \morph{} from $M(B,\es,\mu)$  to $M(B',\es,\mu')$. 

Note that $M(B',\cX,\es,\mu')$ is a submonoid of $M(B,\cX,\es,\mu)$
and $\eps$ realizes the inclusion of these free monoids. Hence
$W= \eps(W') =W'$ as words. In particular, $\sig(W) = \sig ({\ov W})$ implies $\sig'(W') = \sig' ({\ov W'})$. Thus, $(\alp \eps, \sig')$ solves 
$V'$. 

Finally, by definition of $\bet$ we have $\alp = \alp \bet$ because $\alp$ is an $A$-\morph.  Hence $\alp = \alp \eps \bet$ and we obtain $$\alp \eps \sig'(W')= \alp \eps \sig'(W)= \alp \eps \bet \sig(W)
= \alp  \sig(W).$$
\qed
\end{proof}

\subsection{Compression.}\label{sec:compressionSection}
This section finishes the proof of \prref{thm:freecentral}. Consider  an initial 
vertex  
$V_{0} = (\Winit, A, \OO,\es ,\muinit)$ 
with a \solu  $(\alp,\sig)$. We will show below that $\cG$ contains a 
path $V_{0}\arc{h_1} V_1 
\cdots \arc{h_{t}} V_t$ to some final vertex 
$V_t=(W',B,\es,\es,\mu)$ such that 
$\sig(\Winit) = h_{1}\cdots h_{t}(W')$, and so $\cG$ contains all solutions to $\Winit$. Let us show  why then, indeed, we are almost done with  \prref{thm:freecentral}.
We augment the graph $\cG$ by one more vertex which is 
just the symbol $\#$. Recall that $\os{X_{1}\lds X_{k}}$ has been the set of specified variables. Every final vertex $(W',B,\es,\es,\mu)$
has a unique factorization 
$W'=\#w'\#w''$ with 
$\abs{w'}_{\#} = k$. Let us add arcs $(W',B,\es,\es,\mu) \arc {g_{w'}}
\#$ where $g_{w'}:C^*\to C^*$ is the \hom 
(not necessarily respecting the involution) defined by $g_{w'}(\#) = w'$.
If we define the NFA $\cA$ as $\cG$ with this augmentation and if we let 
$\#$ be the exclusive final vertex, then by \prref{prop:backandforth} we obtain 
 \prref{thm:freecentral}.
The construction of  $\cA$ can easily be implemented by a 
nondeterministic procedure that uses $\Oh(n\log n)$ space.  
In order to show the existence of the path from $V_0$ to $V_t$ we apply the recompression method\footnote{Compression became a main tool for solving word equations thanks to \cite{pr98icalp}.}
of \cite{DiekertJP2014csr}, but with a new and improved treatment of 
``block compression''. We avoid solving linear \Dio 
equations,  and give a structural theorem involving 
EDT0L languages that is more precise than the result in \cite{DiekertJP2014csr}.

We show the existence of the path corresponding to the \solu  $(\alp,\sig)$ using an alternation between ``block compression'' and ``pair compression'', repeated until we reach a final vertex. 
The procedures use 
knowledge of the solution being aimed for.
We proceed along arcs in $\cG$ of the 
form $V= (W,B,\cX,\theta, \mu) \arc h V'=(W',B',\cX',\theta', \mu') $ thereby transforming a solution 
$(\alp, \sig)$ to $V$ into a solution $(\alp', \sig')$ to $V'$. However, this is not allowed to be arbitrary: we must keep the invariant $\alp\sig (W) = \alp'h\sig'(W')$.
For example, consider the alphabet reduction where
$B'\varsubsetneq  B$ and $W= W'\in (B'\cup \cX)^*$. In this case 
we have $h = \id{C^*}$, which induces the inclusion  
$\eps: M(B',\es,\mu')\to M(B,\theta,\mu)$. If $\sig$ does not use letters outside $B'$ there is no obstacle.  In the other case, fortunately, we will need alphabet reduction only when the type relation is empty on both sides.   Then we can 
define $\bet(b) = \alp(b) \in A^*$ for $b\in B\sm B'$ and $\bet(b) =b$
for $b\in B'$. We let $\sig'(X) = \bet \sig(X)$. This defines a 
$B'$-solution at $V'$. In some sense this is a huge ``decompression'' 
making $\sig'(W)$ perhaps much longer than $\sig(W)$.
However, $(\alp \eps, \sig')$ is a solution to $V'$. 

A word in $w \in\Sig^*$ is a sequence of {\em positions}, say $1,2\lds \abs{w}$, and each position is labeled by a letter from $\Sig$. 
If $W= u_{0}x_{1}u_{1} \cdots x_{m}u_{m}$, with $u_{i}\in C^*$ and 
$x_{i}\in \OO$, then $\sig(W)= u_{0}\sig(x_{1})u_{1} \cdots \sig(x_{m})u_{m}$ and the positions in $\sig(W)$ corresponding to the $u_{i}$'s are henceforth called {\em visible}.

\subsubsection{Block compression}\label{sec:block}
Let $V= (W,B,\cX, \es, \mu)$ be some current non-final vertex
with an empty type relation and a \solu $(\alp,\sig)$. 
We 
start a block compression only if
$B \leq \abs W\leq \bcc$. Since $|C|=100n$, there will be sufficiently many 
``fresh'' letters in $C\sm B$ at our disposal.

\begin{enumerate}
\item Follow arcs of type $\df 4$ and $\df 6$ to remove all variables with $\abs{\sig (X)} \leq 2$. Thus, without restriction, we have $\abs{\sig (X)} > 2$ for all $X$. 
If $V$ became final, we are done and we stop. Otherwise,
for each $X$ we have $\sig(X)= bw$ for some $b \in B$ and $w \in B^+$. 
Following a substitution arc of type $\df 6$,  we replace $X$ by $bX$. 
(Of course, we also replace 
$\ov X$ by $\ov X \ov b$, changing $\mu(X)$ to $\mu(X) = \ov {\mu(\ov X)}= \mu(w)$. From now on we always do this without further comment.)
Every substitution $X \mapsto bX$ decreases $\sum_{X\in \cX}\abs{\alp \sig(X)}$, a fact which will be used later. Moreover,
if $bX\leq W$ and $b'X\leq W$  are factors with $b,b'\in B$, then $\#\neq b=b'$ due to the previous substitution $X \mapsto bX$.
For each $b\in B\sm \os \#$ 
define sets
$\Lambda_b \sse\N$ which contain those $\lam \geq 2$ such that
 there is an occurrence of a factor $db^\lambda e$ in $\sig(W)$ with $d\neq b\neq e$, where at least one of the $b$'s is visible. We also let $\cX_b=\set{X\in\cX}{bX \leq W \wedge \sig(X)\in bB^*}$.
Note that $\sum_{b} \abs{\Lambda_b} + \abs{\cX_b} 
\leq |W|$.
Since $W$ is well-formed we have $\Lambda_b = \Lam_{\ov b}$.
\item Fix some subset 
$B_+\sse B$ such that for each $\# \neq b \in B$ we have $b \in B_+ \iff 
\ov b \neq B_+$.
For each $b \in B_+$, where $\Lambda_b\neq \es$, run  the following {\em $b$-compression}:
\item{$b$-compression.} (This step removes all proper factors $b^\ell$ and $\ov{b}^\ell$, $\ell \geq 2$, from $W$.) 
\begin{enumerate}
\item Introduce fresh letters $c_{b}, \ov{c_{b}}$ with $\mu(c_{b})= \mu(b)$. In addition, for each 
$\lam \in \Lambda_b$ introduce fresh letters  $c_{\lam,b}, \ov{c_{\lam,b}}$ with $\mu(c_{\lam,b})= \mu(b)$. We abbreviate $c= c_{b}$, $\ov c= \ov{c_{b}}$, $c_{\lam}= c_{\lam,b}$, and $\ov{c_{\lam}}= \ov{c_{\lam,b}}$.
We let 
$h(c_{\lam}) = h(c) = b$ and we introduce a type 
by letting $\theta = \set{(c_{\lam},c)}{\lam \in \Lam_{b}}$.
Renaming arcs~$\df 1$ realize this transformation. We did not touch $W= h(W)$, but we introduced partial commutation. 
\item When we introduced $c, c_{\lam}$ we did not change $W$, but we changed the alphabet $B$ to some larger alphabet $B'$. Now we change 
$W$ and its solution. We start to replace in $\sig(W) \in B^*$ every 
factor $db^\lambda e$ (resp.~$d{\ov b}^\lambda e$), where $d\neq b\neq e$ and $\lam \in \Lam_{b}$, with $dc^\lambda e$ (resp.~$d{\ov c}^\lambda e$). 
This yields a new word $W' \in {B'}^*$, which was obtained via the renaming arc $h(c) = b$. 
Recall that for every $X \in \cX_{b}$ we had $bX \leq W$ and
for some positive $\ell$ we had $\sig(X) = b^\ell w$ with $w \notin bB^*$. In the new word $W'$ we have $cX \leq W'$ and
for the new solution $\sig'$ we have $\sig(X) = c^\ell w'$ with $w'\notin cB'^*$. 
We rename $W',B',\alp' = \alp h, \sig'$ as $W,B,\alp, \sig$.
\item  We define  $\theta = \set{(c_{\lam},c)}{\lam \in \Lam_{b}}\cup \set{(X,c)}{X\in \cX_{b}\wedge \sig(X) \in c^*}$.
  This can be realized by arcs $\df 5$. 
\item Let $W\in M(B,\cX,\theta,\mu)$ be given by some word in $W\in (B\cup \cX)^*$: scan the word $\sig(W) \in B^*$ from left to right. Stop at each 
factor $dc^\lam e $ with $d\neq c \neq e$ and $\lam \in \Lam_{b}$. If in this factor some 
position of the $c$'s is visible then choose exactly one of these visible positions and replace that $c$ by $c_{\lam}$. 
If no $c$ is visible, they are all inside
some $\sig(X)$; then choose any $c$ and replace it
by $c_{\lam}$. Recall that $c$ and $c_{\lam}$ commute, hence
$dc^\lam e$ became $dc_{\lam}c^{\lam-1} e= dc^{\ell_{1}}c_{\lam}c^{\ell_{2}} e \in M(B,\theta,\mu)$ for all $\ell_{1}+\ell_{2}= \lam-1$. 
After that we run through the same steps for $\ov c$. 
The whole transformation can be realized by  renaming arcs $\df 1$ defined by $h(c_{\lam}) = c$. 
There is a crucial observation: if $X\in \cX_{b}$ and 
we had $\sig(X) = c^\ell w$ with $w \notin cB^*$ before the transformation then now still $\sig'(X) = c^\ell w'$, 
but due to commutation $c_{\lam}\sig'(X)$ is a factor in $\sig'(W')\in M(B,\theta,\mu)$. For example, assume $\ov X\ov{c}^2YcZc X \ov c\ov Y d \leq W$ 
with $\sig(X) =cd\ov c$, $\sig(Y) = \ov cdc$, and $\sig(Z) = c^2$.
Then the corresponding factor in $W'$ looks 
as $\ov X\ov{c_{4}}\, \ov cYc_{6}Zc X \ov{c_{4}}\ov Y d
= \ov X\ov{c_{4}}\, \ov cYc_{6}cZ X \ov{c_{4}}\ov Y d \in M(B,\cX,\theta,\mu)$, but $c_{6}$ and 
$Z$ do not commute. However, it is important only that
$\sig'(Z) = c^2$ and  $c_{6}$ commute, which they do. 
\item Rename $W',B',\alp' = \alp h, \sig'$ as $W,B,\alp, \sig$.
Perform the following loop \ref{loopaa} -- \ref{loopda} until no $c$ and no $X\in \cX_{b}$ with 
$\sig(X)\in c^*$ occurs in $W$. 
\begin{enumerate}
\item If $X\in \cX_{b}$ and the maximal $c$-prefix of $\sig(X)$ is odd then follow an substitution arc $X \mapsto cX$. Do the same for $\ov b$ and $\ov c$.
\label{loopaa}
\item If in the new solution $\sig'(W')$ there is a factor 
$dc_{\lam}c^\ell e$ with $\ell$ odd, then inside this factor the 
word $c_{\lam}c$ is visible due to commutation and the previous step. In this case follow an arc $\df 2$ defined by 
$h(c_{\lam})= cc_{\lam}$. Thus, w.l.o.g{} $\ell$ is even for all $dc_{\lam}c^\ell e \leq \sig'(W')$. 
\label{loopba}
\item Follow a compression arc defined by 
$h(c)= c^2$.
\label{loopca}
\item Remove all $X$ with $\sig'(X) =1$ by following an substitution arc; and rename $W',B',\cX',\alp' = \alp h, \sig'$ as $W,B,\cX,\alp, \sig$
\label{loopda}
\end{enumerate}
\item \label{num5}
Let $B'= B\sm\os{c,\ov c}$ and  $\mu'$ be induced by $\mu$. Observe that no $c$ or $\ov c$ appears in $\sig(W)$: they are all compressed into single letters $c_{\lam}$. Thus the type relation of $B\cup \cX'$ is empty again. Hence we can follow an alphabet reduction arc $(W,B,\cX,\theta,\mu) \arc \eps (W,B',\cX',\es,\mu')$. The new solution to $(W,B',\cX',\es,\mu')$ is the pair 
$(\alp',\sig)$ where $\alp'= \alp \eps$ is defined by the restriction 
of $\alp$ to $M(B',\es,\mu')$.
\end{enumerate}
\end{enumerate}
Having performed  $b$-compressions for all $b \in B_+$, we have increased the length of $W$. But it is not difficult to see that the 
total increase can be bounded in $\Oh(n)$.  Actually, we have
$\abs W \leq \pcc$ at the end because we started with $\abs W \leq \bcc$ and step 1 of block compression increases $\abs W$ by at most $2n$, and no other step increases $|W|$.
Now we use alphabet reduction in a final step of block compression in order to reduce the alphabet $B$ such that $\abs B \leq \abs W$. 
We end up at a vertex named again $V= (W,B,\cX,\es,\mu)$, which has a solution
$(\alp, \sig)$. The difference to the situation before block compression is that now $\abs B \leq \abs W \leq \pcc$, and no proper factor $b^2$, $b \in B$, can be found in $W$ anymore.

\subsubsection{Space requirements for the block compression}\label{sec:bcrev}
We start the block compression at a vertex
$V=(W,B,\cX, \es, \mu)$ with a given solution  $(\alp,\sig)$ only if
$\abs B \leq \abs W \leq \bcc$. For example, every initial vertex having a  
solution $(\id{A^*},\sig)$ falls into that category. 
Now we recall all steps. 

First we removed variables $X$ with $\abs {\sig (X) }\leq 2$, and for the 
remaining  variables we did some \subst $X \mapsto bX$, which together increases the 
length by at most $2n$. Hence we reached a vertex $V'=(W',B,\cX', \es, \mu')$
with $\abs {W'} \leq \pcc$ via arcs satisfying the forward property by \prref{lem:ftaus}. Inspecting the procedure step by step, we verify whether 
each time we follow an arc if it satisfies the forward property using Lemmas~\ref{lem:ftaus},~\ref{lem:ftauarcs} and \ref{lem:compiarcs} 
(but without using \prref{lem:alpred}
during the block compression). So wherever we stop, 
Equation~(\ref{eq:forinv}) is valid during this compression procedure if it was valid before. 

Next we have to show termination and also that we stay inside the graph $\cG$ during the procedures. 
Termination can be seen as follows. Renaming arcs are used at most 
$\Oh(n)$ times, so they are irrelevant. 
No arc increases the sum
$\sum_{x\in \cX} \abs{\sig(X)}$. Either it decreases this sum, or
if this sums remains stable, then $\abs W$ decreases. 
This shows termination.

In order to show that we remain inside $\cG$ we have to control 
the possible fluctuations of $\abs W$ and $\abs B$.

After step 1 we have $\abs B\leq \abs W\leq 31n$. 
Step 3(a) increases $\abs B$ by introducing fresh letters $c_b,\ov{c_b}$ plus $\Lambda_b$. Since each such new letter corresponds to a visible letter in $W$, after this step we have $\abs B\leq 62n$. No other steps increase $\abs B$, and at the end we have $\abs B\leq abs W$.

Steps 2, 3(a)-(d) and 3(f) do not change the length of $W$. The only possible increase  to $\abs W$ can occur in step 3(e).
Part 3(e)i may cause an increase of at most $2n$. 
Part 3(e)iii does not increase length, and part 3(e)iii decreases the length of $c,\ov c$-blocks by half. We now show that at step 3(f) we have $\abs W\leq 31n$ again.

Let $\Theta_t=\{c^\dag_1,\dots, c^\dag_r\}$ be the set of letters $c$ or $\ov c$ that are added to $W$ in the $t$-th iteration of step  3(e)i, so $r\leq 2n$.
After step 3(e)i either \begin{enumerate}
\item there is a distinct letter $c^\ddag_i\in\{c,\ov c\}$ for each $c^\dag_i\in\Theta_t$ so that $c^\dag_ic^\ddag_i\leq W$, in which case the compression $h(c)=cc$ in step 3(e)iii will replace $c^\dag_ic^\ddag_i$  by $c^\ddag_i$,
\item there are factors in $W$ of the form $dc_\lambda \alpha c^\dag_{i_1}\dots c^\dag_{i_s}e$ where $d,e\not\in\{c,\ov c\}$ and $\alpha\in\{c^\varepsilon,\ov c^\varepsilon \mid 0\leq \varepsilon<s\}$ (so each $c^\dag_{i_j}$ does not have a distinct $c,\ov c$ letter matching it). In this case each $c^\dag_{i_j}$ came from replacing a variable $X_j$ with $\sig(X_j)=c^\dag_{i_j}$. Some of the $c^\dag_{i_j}$ letters will be consumed by steps 3(e)ii and iii, but if $s-\varepsilon\geq 2$ then some will remain. However for each such letter we have one less variable in $W$, so the number of such letters that can possibly be carried to the next iteration of step 3(e) during the entire procedure is bounded by $2n$.

\end{enumerate} 
So during each iteration of the loop, the number of letters increases in step  3(e)i by at most $2n$, and the number of $c^\dag$ letters that remain after an iteration is at most $2n$, so $\abs W\leq 31n+4n=35n$ throughout step 3(e). At the end of step 3(e) all $c,\ov c$ letters have been consumed, which includes the letters $c^\dag$ added in step 3(e)i, so  $\abs W\leq 31n$.

Hence we end the procedure with an increase in size which comes only from the first step: we end at some vertex
$V'=(W',B',\cX, \es, \mu)$ where the type relation is empty again. Although 
$B'$ might be larger than $B$, we have $\abs {B'} \leq \abs {W'} \leq \pcc$.

\subsubsection{Pair compression}\label{sec:pair}
After one round of block compression we run Je\.z's procedure {\em pair compression}. It brings us back to $\abs B \leq \abs W \leq \bcc$ and allows us to start another block compression. This is essential because it keeps the length in $\Oh(n)$. 

Let us explain the procedure in detail. This is similar to \cite{jez13stacs}, but  some of the technical details are different. 
We roughly follow \cite{DiekertJP2014csr}.

We start a pair compression at a vertex $V_{p}=(W,B,\cX,\es,\mu)$ with 
$\abs B \leq \abs W \leq \pcc$, where $W$ contains no proper factor $b^2$ for any 
$b\in B$. Thus we can assume to have just performed a block compression.  Moreover, we assume that $(\alp, \sig)$ is a \solu to $V$.
The goal is to end at a vertex $V_{q}=(W'',B'',\cX',\es,\mu''')$
with $\abs {B''} \leq \abs {W''} \leq \bcc$ by some path satisfying the forward condition. 

We begin with a random 
partition $B\sm \os \# = L \cup R$ such that 
$b\in L \iff \ov b \in R$ which is constructed as follows. 
We first write $B\sm \os\# = B_{+} \cup \set{\ov b}{b\in B_{+}}$ as a disjoint union. This is possible because 
$B\sm \os\#$ has no self-involuting letters. Next,  
for each $b\in B_{+}$ we choose uniformly and independently whether either 
$b\in L$ and $\ov b \in R$, or $b\in R$ and $\ov b \in L$. 

We have $ab\in LR \iff \ov b \ov a\in LR$, hence the compression respects the involution, and there is no overlap in $\sig(W)$ between any occurrences of $ab\in LR$ and $cd\in LR$, unless it is the same occurrence and $ab= cd$. The idea is to compress in one phase all factors $ab\in LR$ into a single fresh letter.  But the obstacle is that in $\sig(W)$ an occurrence of some $ab \in LR$ can be ``crossing''. This means, there is some  $aX \leq W$ (resp.~$Xb \leq W$) with $\sig(X) \in bB^*$  (resp.~$\sig(X) \in B^*a$). In this case, $\sig(W)$ has an occurrence of a factor $ab$ 
 where $a$ is visible, but $b$ is not visible. This forces us to ``uncross'' $ab$, which is a basic idea from \cite{jez13stacs} and appears in steps (2.) and (3.) below.

\begin{enumerate}
\item Remove all $X\in \cX$ with $\sig(X) = 1$ 
 via \subst arcs. 
\item
Make $B'$ large enough such that $B\sse B' \sse C$ and for each $ab\in LR$ there exists a uniquely defined ``fresh'' letter $c_{ab}\in B'\sm B$, subject to the condition that first, $ab \leq \sig(W)$  and second,  there is at least one occurrence of that factor such that either the position of $a$ or $b$ (or both)
in $\sig(W)$ is visible. 
Let us count how many fresh letters $c_{ab}$ we need. 
There are at most $\pcc$ positions in  $\sig(W)$ which are visible
since $\abs W \leq \pcc$. Moreover, each visible position leads to  at most 
one factor $ab\in LR$. Thus, the number of  $c_{ab}$ is less than 
$\pcc$. Now, for each $c_{ab}$ we let $\ov{c_{ab}} = c_{\ov{b}\ov{a}}$. 
Let us argue that $c_{\ov{b}\ov{a}}$ is present. This is clear
if $ab \leq W$ because in that case $\ov{b}\ov{a}\leq W$, too. 
Otherwise some $a$ is visible and $aX \leq W$ for a variable $X$. We obtain $\ov X \ov a \leq W$, so the position 
of $\ov a$ is visible in $\sig(W)$, and, hence,  
$\ov{b}\ov{a}\leq \sig(W)$ creates the letter $c_{\ov{b}\ov{a}}$.
Thus, $\abs{B'\sm B} \leq \pcc$ and therefore $\abs {B'} \leq {62n}$.

The following always holds: $ab \neq \ov b\ov a$. If we assume the contrary, then 
a factor $ab = a \ov a$ appears in $\sig(W)$ and suppose $a$ is visible.  
We cannot have $a\ov a \leq W$ because in a well-formed word there are no proper factors $w$ with $\mu(w) = 0$. Hence, (by symmetry) we must have $aX \leq W$ for some $X$ with 
$\sig(X) = \ov a w'$ and, as a consequence, $\mu(X) = \mu(\ov a) \mu(w')
= (\ov a, \ov a) \cdot \mu(w') \in N$. Thus, $\mu(aX) =0$, but $aX$ is a proper factor of $W$, contradiction. 
Note the rather far reaching consequence of this last tiny computation: 
as $ab \neq \ov b \ov a$ we can compress later unambiguously $ab$ into $c_{ab}$ and $\ov b \ov a$ into $\ov {c_{ab}}$ without creating any self-involuting letter. Thus,  $c_{ab}\neq \ov {c_{ab}}$ and we maintain the invariant that no other symbol than $\#$ is self-involuting. 
 
We realize this alphabet enlargement, from $B$ to $B'$, via compression  arcs labeled by $h(c_{ab}) = ab$. 
So far we have not changed $W$. We simply followed compression arcs
which satisfy the forward property. 
We are now at some vertex
$V=(W, B',{\cX'}, \es, \mu')$ with the solution 
$(\alp h, \sig)$. Indeed,  $\alp h \sig(W) = \alp \sig(W)$ as $\sig(W) \in B^* \sse B'^*$. 
\item Create a list $\cL= \set{X\in {\cX'}}{\exists b \in R: \sig(X)\in bB^*}$.
For each $X\in \cL$ do in any order: 
\begin{itemize}
\item if $\sig(X)\in bB^*$ with $b \in R$ then follow a \subst arc $X \mapsto bX$. 
\end{itemize}
Remember, if we follow $X \mapsto bX$ then automatically 
$\ov X$ is replaced with $\ov X \, \ov b$, too; and $\ov b \in L$. 
We also have $\os{X,\ov X} \sse \cL$ \IFF $\sig(X) \in bB^*a$ for some $ab\in LR$.
In that case we actually substituted $X$ by $bXa$ and $\ov X$ by $\ov a \ov X\,  \ov b$. In any case, we have successfully ``uncrossed'' every $ab \in LR$. We followed \subst arcs which satisfy the forward property by 
\prref{lem:ftaus}. We are now at a vertex $V'=(W', B',{\cX'}, \es, \mu'')$
with a \solu $(\alp h,\sig')$.
If $ab \leq \sig'(W')$ then in every occurrence of $ab$ either both positions $\sig'(W')$ in are visible or neither are. This concludes the ``uncrossing''.
\item For each $ab\in LR$ such that $c_{ab}$ was generated follow a compression 
arc labeled by $h(c_{ab}) = ab$: 
we replace all occurrences of $ab$ in $\sig'(W')$ by the letter $c_{ab}$,
and simultaneously replace all occurrences of $\ov{b}\ov{a}$ in $\sig'(W')$ by the letter $\ov{c_{ab}}$, with no overlapping ambiguity. It is therefore clear that the \solu $\sig'$ factorizes through $h$. 
Due to the previous ``uncrossings'', all arcs satisfy the forward property. Note that $h^2 =h$ because
$h: B' \to B'$ is a $B$-\morph and $a, b \in B$. Hence, we are now at a 
vertex  $V'=(W'', B',{\cX'}, \es, \mu''')$ with 
$W' = h(W'')$ which has a \solu $(\alp h,\sig'')$ where $\sig''= h \sig'$. 
\item Finally, we perform an alphabet reduction as 
in $\df 3$ which replaces $B'$ by some smaller alphabet $B''$. 
As $\#a\# \leq W''$ (we never touched any occurrence of $\#a\#$)
we see that automatically $A\sse B''$. The alphabet reduction is done 
when the type relation is empty. Thus, we use \prref{lem:alpred} to see that 
we can realize this step by an arc which satisfies the forward property.
We achieved our goal of reaching a vertex $V_q=(W'', B'',{\cX'}, \es, \mu''')$
with a \solu of the form $(\alp h,\bet\sig'')$.
 This finishes the procedure {\em pair compression}.
 \end{enumerate}

The termination of the procedure is immediate. The maximal number of steps is bounded by $\Oh(n)$. 
Suppose $V_{p}=(W,B,\cX,\es,\mu)$ is the start vertex and $V_{q}=V''= (W'', B',{\cX'}, \es, \mu''')$ is the endpoint. Denote this path by 
$$V_{p} \arc {h_{p+1}} V_{p+1} \arc {h_{p+2}} V_{p+2} \arc {h_{p+3}}\cdots 
\arc {h_{q}} V_{q}.$$ 
Along the path we used only arcs satisfying the forward condition. 
 Thus, starting with a \solu $(\alp,\sig)$ we find some \solu
 $(\alp {h_{p+1}}\cdots {h_{q}}, \sig'')$ to $V_{q}$ such that
 $$\alp\sig(W)= \alp {h_{p+1}}\cdots {h_{q}} \sig''(W'').$$
 Thus we maintained the validity of Equation (\ref{eq:forinv}) throughout all iterations
of block and pair compressions.

What remains to be shown is that we can achieve $\abs {W''} \leq \bcc$
 for at least one partition $B\sm \os \# = L \cup R$.

We reformulate the probabilistic argument of \cite{jez13stacs} in our setting.
We started a pair compression at a vertex $V_{p}=(W,B,\cX,\es,\mu)$ with $\abs W\leq \pcc$.
Let us  factorize  the word $W \in (B \cup \cX)^+$ as 
$W = x_{0}u_{1}x_{1}\cdots u_{m}x_{m}$
 such that
\begin{enumerate}
\item $u_{i}\in (B\sm \os \#)^+$ for $1 \leq i \leq m$, \ie each $u_i$ is a nonempty word
over constants,
\item The length of each $\abs {u_{i}}$ is divisible by $3$,
\item $x_i \in (B \cup \cX)^*$ for $0 \leq i \leq m$,
\item $\abs{x_{0}\cdots x_{m}} \leq 3n$.
\end{enumerate}
This is possible because ${\abs W}_{\#}  +\sum_{X\in \cX}{\abs W}_X \leq n$. We need $\abs{x_{0}\cdots x_{m}} \leq 3n$ rather than $n$ because we must adjust the lengths of  the 
$x_{i}$'s in order to guarantee divisibility by $3$ of the $\abs {u_{i}}$'s. By inserting factors of the form $x_{i}=1$
we may assume:
\begin{align}\label{eq:loW}
\abs {u_{i}}&= 3 \quad \text{ for all $1 \leq i \leq m$,}\\
\abs W &= \abs{x_{0}\cdots x_{m}} + 3m.\label{eq:loWWW}
\end{align}
Consider the word $W'$ which was obtained by the \subst arcs, but before the compression of factors $ab\in LR$ into single letters. 
The increase in length $\abs{W'} - \abs{W}$ is due to
substitution arcs  $X\mapsto bX, \ov X\mapsto \ov X \, \ov b$ with $X \in \cL$, so the length goes up by $2n$. Note that the $u_i$ factors do not change, only the $x_i$ factors.
Hence, $W'$ has the factorization 
$W' = y_{0}u_{1}y_{1}\cdots u_{m}y_{m}$ with  $y_i \in (B \cup \cX)^*$ and  
\begin{align}\label{eq:loWw}
\abs{y_{0}\cdots y_{m}}\leq \abs{x_{0}\cdots x_{m}} +2n.
\end{align}

Let $W''$ be the word after pair compression has been performed. So $\abs{W''}\leq \abs{W'}$.
If $\abs W \leq 27n$ 
then
\[\abs {W''}  \leq \abs{W'} =|y_0\dots y_m|+ |u_1\dots u_m| \leq |x_0\dots x_m|+2n+|u_1\dots u_m|\]\[=|x_0u_1\dots u_mx_m|+2n=|W|+2n\leq 27n+2n=\bcc\]
and we are done.

Hence, let us assume $27n \leq \abs W$. We have $27n \leq\abs W=|x_0\dots x_m|+3m\leq 3n+3m$ which means
 $m \geq 8n$.

The word $W''$ is the compression of a
 word 
$y_{0}v_{1}y_{1}\cdots v_{m}y_{m}$
where each $v_{i}$ is the result of the compression restricted to $u_{i}$.
Each $u_{i}$ can be written as $u_{i}= abc$ with $a,b,c \in B$. 
Since $W$ did not contain any proper factor $d^2$ with 
$d \in B$, (as we have performed block compression first) 
we know $a \neq b \neq c$. 
There are two possibilities: either $b\in L$ or 
$b\in R$. 
In the first case, either $c\in R$ or $c\in L$, and in the second case either $a\in L$ or $a\in R$. Each event 
$bc\in LR, bc\in LL, ab\in LR, ab\in RR$
has probability $\frac14$, so with probability  $\frac12$ $u_i$ is compressed from length 3 to 2, and with probability  $\frac12$ it remains length 3.
Thus, for the expectation we obtain 
$E[\abs {v_{i}}]=\frac22+\frac32=\frac52$. 
By linearity of expectation, we obtain
\begin{align}\label{eq:looW}
E[\abs{v_{1}\cdots v_{m}}] = \tfrac{5}{2}m.
\end{align}
Since the expected length is $\frac52m$, and $m$ is even, there must exist 
at least one partition $B\sm \os \# = L \cup R$ which satisfies $\abs{v_{1}\cdots v_{m}}\leq \frac52m$.

So, we ``change our algorithm'' and force the algorithm to choose exactly this partition (in other words, there is a path in the graph which chooses this partition, so we follow this one).
We can thus guarantee $\abs{v_{1}\cdots v_{m}} \leq \frac52m$ by this choice.
We may estimate the length of $W''$  as follows.

\begin{alignat*}{2}\label{eq:loWww}
\abs{W''} &\leq \abs{y_{0}\cdots y_{m}} +\abs{v_{1}\cdots v_{m}}
\\&\leq \abs{x_{0}\cdots x_{m}} +2n  + \tfrac52m 
\\ &= \abs {W} +2n - \tfrac{m}2 &\qquad\text{since 
$\abs{W} = \abs{x_{0}\cdots x_{m}} +3m$}
\\ & \leq \abs {W} - 2n &\qquad\text{ since $m \geq 8n$ }
\\ & \leq \bcc &\qquad\text{ since $|W|\leq \pcc$.}
\end{alignat*}

Due to $\abs{W''} \leq \bcc$, we can run another block compression, then a pair compression and so on. We alternate between 
block and pair compressions inside the graph $\cG$, always following arcs 
satisfying the forward condition.

\subsubsection{Compression terminates}

Let $s(0)=|\Winit|$ and $s(i)$ be the word obtained from $s(i-1)$ by a single of application of block then pair compression.

A priori, although we have $s(i)\leq \bcc$ for all $i\geq 0$ the compression method could run forever.
Let us show that this can never happen.

\begin{lemma}[Termination]\label{lem:compressionTerminates}
The alternation between block and pair compression
terminates with some final vertex. 
\end{lemma}
\begin{proof}
By contradiction assume the contrary.
Then there exists an infinite path alternating between block and pair compressions satisfying the forward condition.
However, due to the first step in the block compression this infinite path
uses infinitely many substitution arcs $X\mapsto bX$, where $b$ is a constant. 
As all arcs satisfy the forward condition for each 
substitution, each use of an arc $X\mapsto bX$ decreases 
$\sum_{x\in \cX} \abs{\alp\sig(X)}$. No use of any arc increases this sum. Since we started with a fixed \solu $(\id{A*},\sig_{0})$ at some initial vertex,
there are no such infinite paths. This is a contradiction, and 
\prref{thm:freecentral} is shown.
\qed\end{proof}

\section*{PART II. Proof of Theorem~\ref{thm:procentral}}
This part contains some repetitions of what has been written above. 
We hope that these redundancies make the reading easier. More importantly, 
the set-up for Theorem~\ref{thm:procentral} is more general, and this generates additional technical arguments.
The proof of \prref{thm:procentral} is therefore more technical and more difficult than that of \prref{thm:freecentral} because we allow more general constraints, 
 and we have to cope with the elements of order $2$ which appear in 
the free products --  as for example in the modular group.

\section{Preliminaries for the proof of \prref{thm:procentral}}\label{sec:basics}

\subsection{Free products: special features of $\F$}\label{sec:fP}
Our results hold for finitely generated free products
$$\F = \star_{1 \leq i \leq p}{F_i}$$
where each $F_i$ satisfies one of the following conditions: 
\begin{itemize}
\item $F_i= \FG{A_i}$ is a free group with basis $A_i$.
\item $F_i= A_i^*$ is a free monoid with involution over a set $A_i$ with 
\invol.
(Recall that the \invol on $A_i$ might be the identity. Hence, if
$A_i^*$ is the free monoid, then the involution means reading words {}from right-to-left.)
\item $F_i$ is a finite group. (Every group is viewed as a monoid with \invol by defining $\ov x = \oi x$.)
\end{itemize}
The monoid $\F$ is a monoid with an involution that is induced by the 
\invol{}s on each $F_i$. It is not essential that the finite monoids 
$F_i$ are groups, but it simplifies the presentation as there are fewer cases.
By $U(\F)$ we denote the submonoid of {\em units}, that is, the invertible 
elements in $\F$. Thus $U(\F)$ is the free product over those $F_i$ which are groups. 
 To make our results nonvacuous we assume that $\F$ is infinite. 
 Also, we assume that the multiplication table for each finite group $F_i$ is part of the input.  

Given $\F$, we choose as a set of monoid generators the smallest subset
$A_\F \sse \F$ which is closed under \invol and which contains the sets 
$A_i$ (if $F_i= \FG{A_i}$ or $F_i= A_i^*$) and all  sets $F_i \sm \os {1}$ where $F_i $ is finite. We let $\pi: A_\F^*\to \F$ be the canonical \morph. 

If $\F$ is a group, then $\F$ is a finitely generated free product 
of infinite cyclic and finite groups. These groups are also known as 
{\em plain groups} or as {\em basic groups}. Basic groups form a proper subclass of the class of {\em virtually free groups} which are those groups
having a free subgroup of finite index.\footnote{For example, the modular group
$\PSL(2,\Z) = \Z/2 \Z \star \Z/3 \Z $ is  plain. The group $\SL(2,\Z)$ is
isomorphic to the amalgamated product $\Z/4 \Z \star_{\Z/2\Z} \Z/6 \Z$. Hence, it is virtually free, but it is not plain because it is infinite and has a non-trivial center.} A geodesic triangle in the Cayley graph of a 
plain group, with respect to the standard generators, is depicted in 
\prref{fig:xyzplain}.

A word $w \in A_\F^*$ is called {\em reduced} if it is the shortest word 
representing the element $\pi(w) \in \F$. (These words are also called 
{\em geodesics} in the literature.)  A word $w$ is reduced (resp.~geodesic) \IFF for each 
factor $ab \leq w$ with $a,b\in \AF$ we have first, 
if $a \in \FG{A_i}$ then $b\neq \oi a$ and second, 
if $a \in F_i$ and if $F_i$ is finite then $b\notin  F_i$. 

We identify $\F \sse \AF^*$ with the regular set 
of reduced words in $\AF^*$. 
Note that every word in a free monoid $A_i^*$ is reduced, so factors of the form $a \ov a$ may appear in reduced words. Moreover, if $F_i$ is finite, then there might be elements $x\neq 1 = x^2$. Hence, the equation 
$X^2 = 1$ may have a non-trivial solution. These are  the main reasons why 
the proof of \prref{thm:procentral} is more involved than the proof of \prref{thm:freecentral}.

The monoid for recognizing the regular subset of reduced words is almost the same finite monoid as defined 
in (\ref{eq:nfm}). We replace (\ref{eq:nfm}) by the
following monoid
 $N_\F = \os{1,0} \cup \AF \times \AF$,  
where $1\cdot x = x \cdot 1 = x$, 
$0\cdot x = x \cdot 0 = 0$, and 
\begin{equation}\label{eq:fpnfm}
\begin{array}{llllll}
(a,b)\cdot (c,d) = \left\{\begin{array}{llllll} 
0 &  & \text{if $bc$ is reduced}\\
(a,d) &&  \text{otherwise}
\end{array}\right.\end{array}
\end{equation}
The \morph $\psi_\F: \AF^* \to N_F$ defined by 
$\psi_\F(a) = (a,a)$ for $a\in \AF$ recognizes $\F \sse \AF^*$. 

Note that the group of units is a rational subset of $\F$. We view 
$U(\F) \sse \F \sse \AF^*$ and we let $U = \set{a \in \AF}{a \in U(\F)}$. 
We use the monoid $\B =  \os{1,0}$ in order to recognize $U \cup \os 1$. 
The recognizing \morph $\psi_U$  maps $\os 1 \cup U$ to $1$ and all other letters to $0$.
Now, consider the monoid 
$$N= (\os{1,0} \cup \AF \times \AF) \times \B.$$
Define $\psi: \AF^* \to N$ by $\psi(a) = (\psi_\F(a),\psi_U(a))$. Hence, 
\begin{equation}\label{eq:fpnfmN}
\begin{array}{llllll}
\psi(a) = \left\{\begin{array}{llllll} 
((a,a),0) &  & \text{if $a \in \AF\sm U$}\\
((a,a),1) &&  \text{if  $a\in U$}
\end{array}\right.\end{array}
\end{equation}
Then $\psi$ recognizes $U(\F)$ and $\F$ simultaneously. 
 This works because 
$$U(\F)= \F \sm \AF^* \cdot (\AF \sm U) \cdot \AF^*.$$
Moreover, if $\AF$ is embedded in any larger alphabet $A$ then, by 
extending $\psi$ to a \hom $\psi: A* \to N$ by $\psi(a) = (0,0)$ for all 
$a\in A\sm \AF$, we obtain that $\psi$ recognizes simultaneously  all 
Boolean combinations of the sets $\os 1$, $U(\F)$, $\F$, and $\AF^*$ inside $A^*$. The monoid $N$ has an efficient representation and $\abs N \leq 2(2+{\abs\AF}^2)$. 

In her thesis Mich{\`e}le Benois proved that the family $\Rat(\F)$ forms an effective Boolean algebra. Her statement in \cite{ben69} 
is for free groups only, but her proof holds  for the free product $\F$, too. 
\begin{proposition}[Benois, 1969]\label{prop:benoisnonclassic}
Let $R \in \Rat(\F)$ be rational and 
$R = \pi(L(\cA))$, where $\cA= (Q,\AF^*,\del, I,F)$ is an NFA with $n$ states over the free monoid $\AF^*$  and $\del \sse Q \times (\AF\cup \os 1) \times Q$.
Then there exists an NFA $\cA'$ with $n$ states satisfying
$$R = \pi(L(\cA'))  \sse L(\cA').$$

In particular, $\Rat(\F)$ forms an effective Boolean algebra. 
\end{proposition}

\begin{proof}
(Sketch) 
For $p,q\in Q$ let  $L(p,q)$ denote the set of words labeling a path in $\cA$ from  $p$ to $q$.
We construct an automaton $\cA'=(Q,\AF^*, \delta', I,F)$ by defining $\delta'$ as follows. Set $\del' = \del$. Repeat the following loop: as long as there are $a,b \in \AF$ such that $ab \neq \pi(ab) = c \in \AF \cup \os 1$  with  $ab\in L(p,q)$ but $(p,c,q)\not\in \del'$, 
replace $\del'$ by $\del'\cup \os{(p,c,q)}$. 

This process takes at most $\abs Q^2(\abs {\AF} +1)$ steps before it terminates; and it produces 
an NFA $\cA'$ as desired. 

In order to show that $\Rat(\F)$ forms an effective Boolean algebra, it is enough to show that it is effectively closed under complementation. Therefore let 
$\cA''$ be an NFA accepting the complement $\AF^*\sm L(\cA')$.  
Since $\F$ is a regular subset of $\AF^*$, 
the set $L(\cA'') \cap \F$ is regular, hence rational. We have $\pi(L(\cA'') \cap \F) = \F\sm R$. Thus, $\F\sm R\in \Rat(\F)$ since $\F\sm R$ is the homomorphic image of a rational set.
 \qed
 \end{proof}
 
\prref{prop:benoisnonclassic} is crucial: it is the (only) justification for our convention that rational constraints $X \in R$  for $\F$ are specified by $X \in \oi \rho(m)$ where $\rho:\AF^* \to N$ is a \hom to a finite monoid $N$  and $m \in N$. Recall that for a given \morph
$\sig: \OO \to \AF^*$ the semantics is $\rho \sig (X) =m$. It is essential to have a one-to-one correspondence between $\Rat(\F)$ and the set
$\set{L \in \Rat(\AF^*)}{L \sse \F}$, and this is induced by the \hom 
$\pi: \AF^* \to \F$. 

\subsection{From $\F$ to the free monoid $\AF^*$ with \invol.}
EDT0L languages are closed under finite unions. Making non-determinstic
guesses (which cover all cases) and  
pushing negations to atomic formulas we may assume without restriction that 
the input $\Phi$ is given a conjunction of 
atomic formulas of either type: $U= V$, $U \neq V$, $X \in \oi \rho(m)$, and $X \notin \oi \rho(m)$. 
Recall that $\rho: \AF^* \to N$ is a \hom to a finite monoid.
Since we may assume that $N$ has at least two elements, we can replace 
each subformula $X \notin \oi \rho(m)$ by $X \in \oi \rho(m')$, where 
$m \neq m'$ since, by definition, $X \in \oi \rho(m)$ refers to an evaluation over $\AF^*$ and not over $\F$. 

Concerning $U= V$ and $U \neq V$, we use standard triangulation and obtain the following equations and  inequalities: 
$X\neq Y$, $X= yz$, and $X =1$ with $X,Y \in \OO$ and $\abs x = \abs y = 1$. Without restriction we can assume that $\oi \rho (1)= \os 1$, hence
we can replace $X =1$ by the constraint $X \in \oi \rho(1)$. Thus
only equations $X= yz$ and inequalities 
$X\neq Y$ remain.

Next, we follow a  well-known procedure for replacing these equations and inequalities over $\F$ by equivalent formulas over $\AF^*$, using equations and rational constraints, only. 
For an equation $X=yz$ we use the following equivalence, which is true for all reduced words $x,y,z \in \F \sse \AF^*$: 
\begin{align*}
x =\pi(yz) \iff& \exists P,Q, R 
\,  \exists a,b,c \in \AF\cup \os{1}:\, R \in U(\F) \wedge \\ &
a = \pi(bc)\wedge
x = PaQ \wedge y = PbR \wedge z = \ov R c Q \wedge R \in U(\F).
\end{align*}
For a ``visual'' proof of this equation see Figure~\ref{fig:xyzplain}. 
 \begin{figure}[t]
\begin{center}
\begin{tikzpicture}[scale=1.6, outer sep=0pt, inner sep = 1pt, node distance = 8pt]
 \node (x) at (0,0) {};   
 
  \node[circle, fill] (d) at (90:0.3) {};   
   \node[circle, fill] (e) at (210:0.3) {};   
  \node[circle, fill] (f) at (330:0.3) {};   
  
 \node[circle, fill] (a) at (90:2) {};
 \node[circle, fill] (b) at (210:2) {};
 \node[circle, fill] (c) at (330:2) {};
 
 \begin{scope}[very thick,decoration={
    markings,
    mark=at position .5 with {\arrow{>}}}] 
    
 \draw[thick, postaction={decorate}] (d) -- node  {} node[above right =1] {$c$} (f); 
 \draw[thick, postaction={decorate}] (e) -- node {} node[below =1] {$a$} (f); 
  \draw[thick, postaction={decorate}] (e) -- node  {} node[above left =1] {$b$} (d); 
      
 \draw[thick, postaction={decorate}] (a) -- node[left= 5]  {} node[right =5] {$\ov R=R^{-1}$} (d); 
   \draw[thick, postaction={decorate}] (b) -- node[below right= 3]  {} node[above left =3] {$P$} (e);
 \draw[thick, postaction={decorate}] (f) -- node[above right = 3]  {$Q$} node[below left=3] {} (c);
 \end{scope}
 
 \node[below left of  = b, node distance = 10pt] {$1$};
 \node[below right of  = c, node distance = 10pt] {$x= yz$};
 \node[above of  = a] {$y$};
\end{tikzpicture}
\caption{Part of the Cayley graph of $\F$:
the geodesic triangle to an equation $x=\pi(yz)$.}
\label{fig:xyzplain}\end{center}
\end{figure}
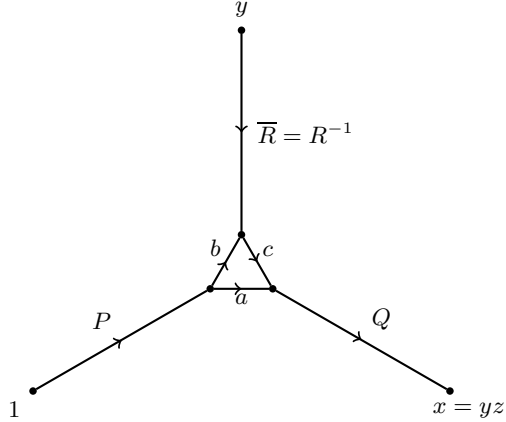

For an inequality $X\neq Y$ we use the following equivalence, which is true for all  words $x,y\in \AF^*$ (it is here that we use the assumption that $\F$ is infinite):
\begin{align*}
x,y\in &\F  \wedge x \neq y \iff \\  &\exists P,Q,R 
\,  \exists a,b,c \in \AF:
b\neq c \wedge
xa = PbQ \wedge ya = PcR \wedge xa\in \F \wedge ya \in \F.
\end{align*}

\begin{remark}\label{rem:whereweare} 
We have shown that it is enough to prove Theorem~\ref{thm:procentral} in the special case where $\F = \AF^*$ is a free monoid with involution. In particular, all words are reduced. The formula $\Phi$ is a conjunction of triangular equations $X=yz$ and of constraints $X \in \oi \rho(m)$. The main remaining obstacle is that $\AF$ might contain self-involuting 
reduced words. In particular, there can be a letter $a$ with 
$a = \ov a$, and  even if there is no such letter then still $a\ov a$ is a self-involuting reduced
word. \end{remark}

\subsection{How to remove self-involuting letters}\label{sec:tick}
The alphabet $\AF$ is an alphabet with involution, denoted here 
by $\wt{\phantom{a}}$. As we are in the general situation 
we might have $\wt a = a$ for some letter $a\in \AF$. 
Choose some subset $\Apos \sse \AF$ such that the size of $\Apos$ is minimal while satisfying $$\AF = \Apos \cup \set{\wt a }{a \in \Apos}.$$ 
By minimality, we have 
$$\Apos \cap \set{\wt a }{a \in \Apos} = \set{a \in \AF}{a = \wt a}.$$
We let $\Aneg$ be a disjoint copy of $\Apos$. We can write $\Aneg = \set{\ov a }{a \in \Apos}$ and then  $\ov {\ov a} =a$ makes $\Apone = \Apos \cup\Aneg$ into an alphabet 
with involution $\mathinvol$ without self-involuting letters. 
Now we encode words over $\AF^*$ as words over $\Apone^*$ via a \morph $\iota:\AF^*\to \Apone^* $, where $\iota(a) = a$ and $\iota(\wt a) = \ov a$ if $\wt a \neq a \in \Apos$, and 
 $\iota(a) = a\ov a$ if $\wt a = a$. Note that $\iota$ respects the \invol. 
 (The \morph $\iota$ is a \emph{code} because $\iota(\AF)$ is the basis of a free submonoid in $\Apone^*$.)
 We also define a \hom $\eta:\Apone \to \AF$ in the other direction 
 by defining $\eta(a)$ and $\eta(\ov a)$ for $a \in \Apos$ as follows. 
 We let $\eta(a) = a$ and if $a \neq \wt a$ then $\eta(\ov a) = \wt a$.
 If, however, $a = \wt a$, then $\eta(\ov a) = 1$. The \hom $\eta$ does not respect the involution, but this does no harm. Note also that $\eta$ erases letters. However, the restriction of  $\eta$ to a subset of $\iota(\AF^*)$ injective.

 We transform $\Phi$ in \prref{rem:whereweare} as follows.
 An equation $U=V$ is replaced by $\iota(U) = \iota (V)$, and 
 every variable $X$ receives an additional constraint $X\in \iota(\AF^*)$.
 Note that $\iota(\AF^*)$ is a regular subset in $\Apone^*$  and we 
 can recognize $\iota(\AF^*)$ by $\rho_{\AF}: \Apone^* \to \N_{\AF}$, where $N_{\AF}$ has size $\Oh(n^2)$. The precise definition of $\rho_{\AF}$ is natural and left to the reader. 
 Previous constraints $X \in \oi \rho(m)$ are replaced by 
 $X \in \oi \eta(\oi \rho(m))$. Thus, we have transformed 
 $\Phi$ over $\AF^* \cup \OO$ into a new formula $\Phi'$ over $\Apone \cup \OO$. Moreover, the construction guarantees that $\eta$ defines a bijection
 between $\cSol(\Phi')$ and $\cSol(\Phi)$. Thus, if 
 $\cSol(\Phi') = \set{\phi(\#)}{\phi \in \L(\cA')}$, then we find another 
 NFA $\cA$ which has just one more state than $\cA'$, and we obtain:
 $$\cSol(\Phi) = \eta(\cSol(\Phi') = \set{\eta \phi(\#)}{\phi \in \L(\cA')}
 = \set{\psi(\#)}{\psi \in \L(\cA)}.$$
 
\begin{remark}\label{rem:wherewearenow} 
By the equation above, \prref{rem:whereweare}, and the last transformations, 
 it is now enough to prove Theorem~\ref{thm:procentral} in the special case where $\F = \Apone^*$ is a free monoid with involution and $\Phi$ is a conjunction of equations and constraints of the form $X \in \oi \rho(m)$. 
 Moreover, $\Apone$ is a set where the involution has no fixed points. 
\end{remark}
\subsection{How to force recognizing \homs to respect the \invol}\label{sec:trick}
Assume the \hom $\rho: \AF^*\to N$ defining the rational constraint in \prref{rem:whereweare} was in fact a \morph to a finite monoid with involution. This property could have been lost during the transformation leading to \prref{rem:wherewearenow}. 

We now replace a  recognizing \hom by a recognizing \morph to a finite monoid with \invol. 
We show that there is a natural embedding of monoids (with or without \invol) into 
monoids with \invol.  If $M$ is any monoid, then we define its dual monoid $M^T$ to be based on the same set $M^T=M$ as a set, where $M^T$ is equipped with a new multiplication $x \circ y= yx$. 
In order to indicate whether we view an element in the monoid $M$ or $M^T$, we use a flag: 
for $x \in M$ we write $x^T$ to indicate the same element in $M^T$. Thus, we can suppress the symbol $\circ$ and we simply write 
$x^T  y^T= (yx)^T$. The notation is intended to mimmick transposition in matrix calculus. 
Similarly, we frequently write $1$ instead of $1^T$ which is true for the identity matrix as well. Now the 
direct product $M\times M^T$ becomes a monoid with involution 
by letting $\ov {(x,y^T)} = (y,x^T)$. Indeed, 
$$\ov {(x_1,y_1^T) \cdot (x_2,y_2^T)} 
= (y_2 y_1, (x_1x_2)^T) 
= \ov {(x_2,y_2^T)}\cdot \ov {(x_1,y_1^T)}.$$

The following items are essential.
\begin{itemize}
\item If $M$ is finite then $M\times M^T$ is finite, too. 
\item We can embed $M$ into  $M\times M^T$ by a \hom $\iota: M \to M\times M^T$ defined by $\iota(x) = (x,1)$. Note that
if $\eta: M\times M^T\to M$ denotes the projection onto the first component, then $\eta \iota = \id{M}$.
In particular, every \hom $\rho: M \to N$ of monoids factorizes through $\iota \rho: M \to N\times N^T$. We have $\rho = \eta \iota \rho$. 
\item If $M$ is a monoid with involution and $\rho: M \to N$ is a \hom of monoids, then 
we can lift  $\rho$ uniquely to a \morph  $\mu: M \to N\times N^T$ of  monoids with \invol such that 
we have $\rho = \eta \mu$. Indeed, it is sufficient and necessary to define 
$\mu(x) = (\rho(x),\rho(\ov x)^T)$.
\end{itemize}
\begin{example}[\cite{dgh05IC}]\label{ex:dgh05IC}
Let $M= \Bn$. Then  $M\times M^T= \Bn\times \left({\Bn}\right)^T$ 
is a submonoid of the set of $2n \times 2n$-Boolean matrices: 
$$\Bn\times \left({\Bn}\right)^T= \set{\vdmatrix P00{Q^T}}{ P,Q \in \Bn} \text { with }  
\overline{\vdmatrix  {P}00 {Q^T}}  =
\vdmatrix {Q}00{P^T}. 
$$
In the line above $P^T$ and $Q^T$ are the transposed matrices. 
\end{example}

\begin{remark}\label{rem:wherewearenowi} 
It is enough to prove Theorem~\ref{thm:procentral} in the special case where $\F = \Apone^*$ is a free monoid with involution and $\Phi$ is a conjunction of equations and constraints of the form $X \in \oi \mu(m)$,
where $\mu: \Apone^* \to N$ is a \morph between monoids with \invol. 
Moreover, $\Apone$ is a set where the involution has no fixed points. 
\end{remark}

\subsection{The initial equation $\Winit$}\label{sec:track}
Recall that, due to \prref{rem:wherewearenowi}, we may assume that 
we are in the special situation where $\F = \Apone^*$  and 
$\Apone$ is a set where the involution has no fixed points.
The next few steps are quite similar to those in the proof of \prref{thm:freecentral}. 
We introduce a special symbol $\#$ with $\ov \# = \#$, and we let 
$A = \Apone \cup \os \#$ be the initial alphabet with involution. 
In this alphabet no other symbol except $\#$ is self-involuting. 

All rational constraints are given by a \morph $\mu_{00}:\Apone^{*}\to N$ where $N$ is a finite monoid with 
involution which has a zero. 
Without restriction we assume that $\oi{\mu_{00}(1)} \sse \os 1$ and that 
for all $w \in \Apone^*$ we have 
$\mu_{00}(w ) \neq 0 \iff w \in \iota (\AF)$, where $\iota$ is the embedding of  $\AF$ into $\Apone$ {}from above. 
We extend $\mu_{00}$ to a \morph $\mu_0: A^* \to N$ by $\mu_0(\#) = 0$. 
We then extend $\mu_0$ to a \morph $\muinit: (A\cup \OO)^* \to N$
by guessing the values for $0 \neq \muinit(X)= \ov{\muinit(\ov X)} \in N$. In the following $\muinit(x)$ changes frequently for symbols outside of $A$, thus we use $\mu,\mu'$ as generic notation. However, the following will hold throughout: 
$\mu_0(a)=\muinit(a) = \mu(a)= \mu'(a)$ for all $a \in A$. The starting point is now a system of equations $U_i= V_i$ with $U_i, V_i \in (\Apone\cup \OO)^*$
for $1 \leq i \leq s$ and the \morph $\muinit: (A\cup \OO)^* \to N$. 
We don't care that $\abs{U_iV_i}= 3$ anymore. 

We let $U'= U_{1}\# \cdots \# U_{s}$ and 
$V'= V_{1}\# \cdots \# V_{s}$. Therefore, we have only one single equation. 
Next, we define the actual initial equation $\Winit$ as in (\ref{eq:Winit}).
\begin{equation}\label{eq:Winiti}
\Winit=  \#x_1\#\dots \#x_\ell \# U'\# V'\#\ov {U'}\# \ov {V'}\# \ov{x_\ell}\#\dots \#\ov{x_1}\#.
\end{equation}
In particular, according to (\ref{eq:Winit})  we have  $\Apone \cup \OO = \os{x_{1}\lds x_{\ell}}$ with $x_{i}= X_{i}$ for $1 \leq i \leq k$.
The overall strategy to proving \prref{thm:procentral} is as before -- we show that the following language is EDT0L:
$$\set{\sig(X_1)\# \cdots \#\sig(X_k)}{
\sig(\Winit) = \sig(\ov{\Winit}) \wedge \muinit= \mu_{0}\sig \wedge 
\forall X: \sig(X) = \sig(\ov{X}) }.$$
 Recall that  for every \morph $\sig: \OO \to \Apone$ we have 
 $$\sig(U')= \sig(V') \iff \sig(\Winit) = \sig(\ov \Winit).$$
 We fix some large enough $n = \ninit  \in \Oh(\abs {\Winit})$ and alphabet
$C$ of size $\Oh(n)$. We assume that $A \sse C$ and no other symbol than $\#\in C$ is self-involuting.
In contrast to the above we content ourselves with 
$\abs C \in \Oh(n)$; we leave it to the reader to calculate large enough constants.

As usual we let $A \sse B , B' \sse C$ and we assume that $B$ 
and $B'$ are closed under involution. By $\cX$, $\cX'$ we denote subsets of $\OO$ which are closed under involution, too. Moreover, we let $\Sig= C \cup \OO$.

The letter $\mu$ refers to a \morph $\mu: (B\cup \cX)^* \to N$ defined
by $\mu: B\cup \cX \to N$ where $N$ is a finite monoid with involution.
We assume that $\mu(a) = \mu_0(a)$ for all $a \in A$. The \morph $\mu$ encodes the rational constraints: in particular, using the appropriate $\mu$ we can guarantee that solutions are in $\iota(\F)$. 


\subsection{Partial commutation induced by a type relation}\label{sec:track}
The definition of type relations appearing in $\cG$ was more general than necessary. For the proof of \prref{thm:procentral} we restrict the type relations 
in order to focus on what we need, because we use now compression arcs 
which have have not been used in \prref{thm:freecentral}. 
The main difference between the proofs of \prref{thm:freecentral} and  \prref{thm:procentral} is that reduced words may have self-involuting factors of the form $a\ov a$ and that compression of $a\ov a$ into a letter would lead to self-involuting letters, which we must avoid. 
The solution is simple: never compress any factor $a\ov a$. It will be enough to compress factors $(a\ov a)^\ell$ for $\ell \geq 2$ down to a self-involuting word of length 2.

We restrict the rather general notion of {\em type relation}
over $\Sig = C\cup \OO$ as follows.  Since it is the restriction of the earlier definition,  $\theta$ is an irreflexive and antisymmetric relation, where $(x,y)\in \theta$ implies $(\ov x,\ov y)\in \theta$. In addition, we require that
first, $(x,y)\in \theta$  implies $y \in \os{c,\ov c, c \ov c}$ for some $c \in C$ and that
there is no $(x',y')\in \theta$  with $x \in \os{c,\ov c}^*$, and second,  
 only the following two forms of type relations $\theta$ are allowed,
where $c\in C\sm \os{\#}$ is a letter. 
\begin{align}
\theta &\sse \set{(x,c), (\ov x, \ov c)}{x \in \Sig\sm \os{c, \ov c}}. \label{eq:I1} 
\\
\theta &\sse \set{(X,c \ov c)}{X \in \OO} \cup \set{(a \ov a,c \ov c)}{a \in C\sm \os{c, \ov c}}. \label{eq:I2}
\end{align}
Moreover, $\abs{\theta(x)}\leq 1$ where $\theta(x) = \set{y\in C^+}{(x,y)\in \theta}$.
Clearly, we maintain  $\abs \theta \in \Oh(n)$, which allows us to store $\theta$ in \qls.  
Given $\theta$ and $\mu: B\cup \cX \to N$ such that 
$\mu(xy) = \mu(yx)$ for all $(x,y) \in \theta$ we define as above the following two partially commutative monoids with \invol. 
\begin{enumerate}
\item $M(B,\cX,\theta,\mu) = (B\cup \cX)^*/\set{xy = yx}{(x,y) \in \theta}$, a monoid with a
\morph $\mu: M(B,\cX,\theta,\mu) \to N$.
\item $M(B,\theta, \mu) = B^*/\set{xy = yx}{(x,y) \in \theta}$, a submonoid of
$M(B,\cX,\theta, \mu)$ such that $$\mu:M(B,\theta, \mu) \hookrightarrow M(B,\cX,\theta, \mu)
\arc{\mu} N.$$
\end{enumerate}
Let us recall that if $w\leq W\in M(B,\cX,\theta, \mu)$, then $w$ is called a {\em proper
factor} if $w \neq 1$ and $\abs{w}_\# = 0$ and that, since  the defining relations for $M(B,\cX,\theta, \mu)$ are of the form $xy=xy$,
we can define $\abs W$ and $\abs{W}_a$ for $W\in M(B,\cX,\theta, \mu)$ by representing $W$ by some word $W\in (B\cup \cX)^*$. 
It follows that we can decide $w\leq W$ in \qls for $w,W \in M(B,\cX,\theta, \mu)$.

As above, $W\in M(B,\cX, \theta, \mu)$ is {\em well-formed} if it is well-formed according to \prref{def:wellf2}.
We repeat the definition of an extended equation and apply it to the restricted version of type relation. 
\begin{definition}\label{def:extequat2}
 An {\em extended equation} is a tuple $V= (W,B,\cX,\theta,\mu)$ 
where $W \in M(B,\cX, \theta, \mu)$ is well-formed. 
A {\em $B$-\solu} of $V$ is a $B$-\morph $\sig:M(B\cup\cX,\theta,\mu)\to M(B,\theta,\mu)$ such that 
$\sig(W) = \sig(\ov W)$ and $\sig(X) \in c^*$ whenever $(X,c)\in \theta$.
A {\em \solu} of $V$ is a pair  $(\alp,\sig)$  such that 
$\alp: M(B,\theta,\mu) \to  A^*$ is an $A$-\morph satisfying $\mu_{0}\alp = \mu$ and $\sig$ is a $B$-\solu.  
\end{definition}


\goodbreak
\subsection{The directed labeled graph $\cG_\F$.}\label{subsec:myG2}

Define the graph $\cG_\F$ to be the induced subgraph of $\cG$ 
which is defined by the set of all extended equations 
$(W,B,\cX,\theta,\mu)$ where $\theta$ satisfies the specification as above. 
(If necessary, adopt the constant $\kappa$ to be large enough, say $\kappa = 100$.)  
In particular, $\theta$ is either of the form (\ref{eq:I1}) or (\ref{eq:I2}).
The restriction is imposed in order to focus on the essential arcs, and
it is allowed to start with $\cG$ as defined originally. In particular, 
we keep the set of {\em initial vertices}, which are the vertices of the 
form $$(\Winit,A,\OO,\es,\muinit).$$ The set of {\em final vertices} 
is again 
$$
\set{(W,B,\es,\es,\mu)}{W= \ov W}.
$$ 
All arcs in $\cG$ that are between vertices of $\cG_\F$ are also arcs in 
$\cG_\F$, since we consider the induced subgraph. In particular,  
\prref{prop:backandforth} still holds, and states the following.

\begin{proposition}\label{prop:backandforth2}
Let $V_{0}\arc{h_1} V_1 
\cdots \arc{h_{t}} V_t$
be a path in $\cG_\F$ of length $t$, where $V_0= (\Winit, A, \OO,\es ,\muinit)$
is an initial and $V_t=(W',B,\es,\es,\mu)$ is a final vertex. 
Then $V_{0}$ has a solution $(\id{A},\sig)$ with
 $\sig(\Winit)= h_1 \cdots h_t(W')$. 
 Moreover, 
 we have 
 $W'\in  \#u_{1}\#\cdots \#u_{k}\#B^*$ such that $\abs{u_{i}}_{\#}= 0$ and 
 we can write:
  \begin{align}\label{eq:wondercX} 
h_1 \cdots h_t(u_{1}\#\cdots \#u_{k}) = \sig(X_{1})\#\cdots \#\sig(X_{k}),
\end{align}
\end{proposition}

\begin{proof}
See \prref{prop:backandforth}.
\qed
\end{proof}

\section{General compression}\label{sec:comp2}
We can now give the proof of \prref{thm:procentral}, following the same scheme as for free groups and the proof of \prref{thm:freecentral}.

Consider  an initial 
vertex  
$V_{0} = (\Winit, A, \OO,\es ,\muinit)$ 
with a \solu  $(\alp,\sig)$. We show that ${\cG}_{\F}$ contains a 
path $V_{0}\arc{h_1} V_1 
\cdots \arc{h_{t}} V_t$ to some final vertex 
$V_t=(W',B,\es,\es,\mu)$ such that 
$\sig(\Winit) = h_{1}\cdots h_{t}(W')$. 
We show the existence of the path using a repetition  
of the sequence:
\begin{center}
``block compression'', ``non-standard block compression'', ``pair compression''.
\end{center}
Let us recall that the scheme is repeated until we reach a final vertex and that the procedures use some external knowledge about solutions. We proceed along arcs $V \arc h V'$ in ${\cG}_{\F}$ thereby transforming a solution 
$(\alp, \sig)$ to $V$ into a solution $(\alp', \sig')$ to $V'$ such that we keep the invariant $\alp\sig (W) = \alp'h\sig'(W')$.
\subsection{Standard block compression}\label{sec:block2}
The procedure has been described above in the main body of the paper as 
{\em block compression}. 
We start at a non-final vertex $V= (W,B,\cX,\es, \mu)$ with a \solu 
$(\alp, \sig)$ with $\abs B \leq \abs W \in \Oh(n)$. We move along arcs 
satisfying the forward condition and we arrive 
at a vertex $V'= (W',B',\cX',\es, \mu')$ with a \solu 
$(\alp', \sig')$. We have $\abs B' \leq \abs W' \in \abs W +\Oh(n)$, so there is a possible increase by $\Oh(n)$ in the length of the equation,
but we know that $W'$ does not contain any proper factor $b^2$ with 
$b \in B'$. At the end of the standard block compression we rename 
$V'= (W',B',\cX',\es, \mu')$ and $(\alp', \sig')$ as $V= (W,B,\cX,\es, \mu)$
$(\alp, \sig)$, but we keep in mind the increase of length by $\Oh(n)$.

After that, we start the non-standard block compression to remove all factors 
$a\ov a a$ from $\sig(W)$. This is explained next. 
\subsection{Non-standard block compression}\label{sec:nsblockc}
We follow the explanation and notation according to the main body of the paper. We consider some non-final vertex $V = (W,B,\cX, \es, \mu)$ with an empty type relation and a \solu $(\alp,\sig)$. 
Let $B\sm \os \# = \Bpos \cup \Bneg$ be any partition such that 
$b\in \Bpos \iff \ov b \in \Bneg$. Recall that $a \ov a$
could be a reduced word. As $a \ov a$  is a self-involuting word we cannot compress it into a single letter $c$ because then the letter $c$ is forced to be self-involuting, 
since compression of a factor $a \ov a = \ov{a \ov a}$ must be unambiguous.
Note that $a \ov a$ has no non-trivial self-overlap since $a \neq \ov a$

Let $c \in C\sm B$ and 
$h:C^*\to C^*$ be the renaming \hom defined by $h(c)=a$ and
$h(\ov c)=\ov a$. 
Let $w \in B^*$ be any word. Then there is a unique word $w' \in B^*$
such that $w = h(w')$ and $w'$ does not have any factor $a\ov a$ though it may have a factor $\ov a a$. 
The word $w'$ can be obtained by replacing every occurrence of 
$a\ov a$ by $c\ov c$. The word $w'$ is unique because $a \ov a$ has no non-trivial self-overlap.

 Let us  introduce the following notation. For  $a \in C$, $w \in \Sig^* = (C \cup \OO)^*$, and $\lam \geq 1$ such that 
$(a \ov a)^\lambda$ is a factor of $w$. 
We  say that an occurrence of $(a \ov a)^\lambda$ in $w$ is {\em maximal} if the occurrence corresponds to a factorization $w = u (a \ov a)^\lambda v$ such that neither $u \in  \Sig^*a \ov a$ nor 
$v\in a \ov a \Sig^*$. This means that at least one occurrence of $(a \ov a)^\lambda$ in $w$ is not contained in any occurrence of a factor 
$(a \ov a)^{\lambda+1}$. 
 Note that $(a \ov a)^1$, $(a \ov a)^2$, $(a \ov a)^3$ may be factors in some $w$, both $(a \ov a)^1$ and  $(a \ov a)^2$
have maximal occurrences, but $(a \ov a)^3$ does not: for example, 
$w = (a \ov a)^1\#(a \ov a)^2\#(a \ov a)^4$. 
Note also that 
$a \ov a$ has a maximal occurrence in $w= \ov a a \ov a a$. 

We can  repeat, partly verbatim, the standard block compression, but there are 
slight modifications, and we divide the procedure into smaller steps.

\begin{remark}
Before we describe the procedure in mathematical terms,  let us try to give a a high level  explanation what  
a non-standard block compression does. The basic idea is simple. 
In order to avoid self-involuting letters we cannot compress $c \ov c$ into a single letter, but we can compress $c \ov c c \ov c$ into the word  $c \ov c$. This means we can compress maximal blocks $(c \ov c)^{2\ell}$ into blocks $(c \ov c)^{\ell}$. 
The compression must correspond to \morph{s}. This is fine: the \morph  $c \mapsto c \ov c$ maps $c \ov c$ to $c \ov c c \ov c$.
But then we can continue the same way only if $\ell$ is even.
Therefore there is some extra work necessary if $\ell$ becomes odd. 
If there is a maximal block $(c \ov c)^{\ell}$ with $\ell$ odd then we replace first  $(c \ov c)^{\ell}$ by $c_\lam \ov{c_\lam} (c \ov c)^{\ell-1}$
where $c_\lam$ is a fresh letter.
Once we have  $c_\lam \ov{c_\lam}$ available, we can compress $(c \ov c)c_\lam \ov{c_\lam}(c \ov c)$ into  $c_\lam \ov{c_\lam}$.
The morphism maps $c_\lam$ to $c \ov c c_\lam$. Thus, $c_\lam \ov{c_\lam}$ is mapped to $(c \ov c)c_\lam \ov{c_\lam}(c \ov c)$  which is equal to $c_\lam \ov{c_\lam}(c \ov c)^2$ due to partial commutation. 
 So at the end every maximal block  $(c \ov c)^{\ell}$ where $\ell$ is even or odd gets compressed 
into some $c_\lam \ov{c_\lam}$. One could say that $c_\lam \ov{c_\lam}$ is a
special pair with a sort of marker which always guesses correctly whether it sits ``inside'' some block $(c \ov c)^\ell$ with $\ell \equiv 2 \bmod 4$  or with $\ell \equiv 0 \bmod 4$. This is why counting $\bmod {} 4$ comes in.\end{remark}

\goodbreak
{\noindent \bf begin non-standard block compression}

We begin at vertex $V = (W,B,\cX, \es, \mu)$ with  a \solu $(\alp,\sig)$.
During the process we introduce partial commutation but it vanishes at the end.
After each transformation we rename the current vertex as $(W,B,\cX, \theta, \mu)$ with  a \solu $(\alp,\sig)$. During the $(a\ov a)$-compression 
we enlarge the alphabet and then the standard notation for a vertex 
becomes $(W,B',\cX, \theta, \mu)$.
\begin{enumerate}
\item Follow arcs of type $\df 4$ and $\df 6$ in order to remove all variables with $\abs{\sig (X)} \leq 10$. Thus, without restriction, we have $\abs{\sig (X)} > 10$ for all $X$. 
If $V$ became final, we are done and we stop. 
\item 
For each $X$ we now have $\sig(X)= bw$ for some $b \in B$ and $w \in B^+$. 
Following a substitution arc $\df 6$, we replace $X$ by $bX$, $\ov X$ by $\ov X \,\ov b$ and change $\mu(X)$ to $\mu(X) = \ov {\mu(\ov X)}= \mu(w)$. 
Now, if $bX\leq W$ and $b'X\leq W$  are factors with $b,b'\in B$ then $\#\neq b=b'$ due to the previous substitution $X \mapsto bX$.

\item 
For each $a\in \Bpos$ 
define sets
$\Lambda_a \sse\N$ which contain those $\lam\geq 1$ such that
there is a maximal occurrence of $(a \ov a)^\lambda$ in $\sig(W)$ where at least one of the $a$'s is visible. Note that we treat $a\in \Bpos$ different from $\ov a \in \Bneg$. However, this is not essential here. 
If there is a maximal occurrence of $(a \ov a)^\lambda$ in $\sig(W)$ where at least one of the $a$'s is visible then there is another maximal occurrence of $(a \ov a)^\lambda$ in $\sig(W)$ where at least one of the $\ov a$'s is visible. 
We have $\sum_{a\in \Bpos} \abs{\Lambda_a} \leq \abs W$.
We also let 
$$\cX_a=\set{X\in\cX}{aX \leq W \wedge \sig(X)\in \ov a B^* 
\vee  \ov aX \leq W \wedge \sig(X)\in a \ov a B^*}.$$
Note that if $X\in \cX_{a}$ for some $a \in \Bpos$ then either the 
factor $a \ov a$ or the factor $\ov a a$ or both have a ``crossing'' in $\sig(W)$. 
\item 
For each $\Lambda_a \neq \es$ -- one after another -- run  the following subroutine, called {\em $(a\ov a)$-compression}. 
The purpose is to remove all proper factors $(a \ov a)^\ell$ with $a \in B_+$ and $\ell \geq 1$ from $W$. More precisely, if $(a \ov a)^\ell$ is a maximal occurrence of that factor in $W$ then
the following $(a\ov a)$-compression replaces this occurrence $(a \ov a)^\ell$ by
some factor $c_\lam \ov{c_\lam}$. We do not have $\ell = \lam$ in general, but clearly 
$\abs{(a \ov a)^\ell} \geq \abs {c_\lam \ov{c_\lam}} = 2$.
\end{enumerate}
{\bf end non-standard block compression}
\subsection*{Subroutine $(a\ov a)$-compression.} 
The subroutine is called at a vertex $V = (W,B,\cX, \es, \mu)$ with \solu $(\alp,\sig)$.

{\noindent \bf begin $(a\ov a)$-compression}
\begin{enumerate}
\item  Introduce fresh letters $c_{a}, \ov{c_{a}}$ with $\mu(c_{a})= \mu(a)$. In addition, for each 
$\lam \in \Lambda_a$ introduce fresh letters  $c_{\lam,a}, \ov{c_{\lam,a}}$ with $\mu(c_{\lam,a})= \mu(a)$. Define $h(c_{\lam,a}) = h(c_a) = a$ and introduce a type
by letting $$\theta= \set{(c_{\lam,a}\ov{c_{\lam,a}},c_a\ov {c_a})}{\lam \in \Lam_{a}}.$$
Renaming arcs $\df 1$ realizes this transformation. We did not touch $W$, hence $W= h(W)$, but we enlarged $B$ to some set $B'$ and we introduced partial commutation. We abbreviate $c= c_{a}$, $\ov c= \ov{c_{a}}$, $c_{\lam}= c_{\lam,a}$, and $\ov{c_{\lam}}= \ov{c_{\lam,a}}$.
\item (Change 
$W$ and its solution.) We replace in $\sig(W) \in B^*$ every 
maximal occurrence of a factor $(a \ov a)^\lam$ with $\lam \in \Lam_{a}$ by $(c \ov c)^\lam$. This yields a new word $W' \in {B'}^*$. The transformation can be realized again by a single renaming arc defined by $h(c) = a$ and leading to a solution $(\alp',\sig')$. 

Various $c$ and $\ov c$ appear in $W'$ and $\sig'(W')$.
Note that every $c$ in $\sig'(W')$ is followed by some $\ov c$ and every 
$\ov c$ is preceded by some $c$. 
Some of their positions are visible. For example, $W'$ may have  factors of the form $X\ov cY$ or $X\ov cc\ov c cY$ etc. 
We rename the vertex and its \solu as $(W,B',\cX,\theta,\mu)$ and $(\alp, \sig)$. 
\item  Consider all $X\in \cX_{a}$ where $\sig(X)$ is a factor of some word 
in $(c\ov c)^*$. (For example, $\sig(X) \in \ov c (c\ov c)^* c$.) 
Follow \subst arcs $\df 6$
such that first, for the resulting \solu $\sig'$ we have
$\sig'(X) \in (c\ov c)^{4m}$ for some $m\in \N$ and second, each occurrence of such $X\in \cX_a$ in $W'$ occurs inside an occurrence of
$c\ov c X$. As usual, we rename the vertex and its \solu as $(W,B',\cX,\theta,\mu)$ and $(\alp, \sig)$. 

\item (Typing variables.) 
Enlarge  $\theta$  such  that it becomes
$$\theta = \set{(c_{\lam}\ov{c_{\lam}},c\ov c)}{\lam \in \Lam_{a}}
\cup\set{(X,c\ov c)}{X \in \cX_{a}\wedge \sig(X) \in (c\ov c)^*}.$$
\item 
For all  $X\in \cX_{a}$ where $\sig(X) \notin (c\ov c)^*$  factorize 
$\sig(X) = uv$ such that $u$ is a suffix (possibly empty) of some word in 
$(c\ov c)^*$ and $v \notin c\ov cM(B',\theta,\mu)$. Following more \subst arcs
we can make sure that first, $u \in (c\ov c)^{4m}$ for some $m\in \N$ and
second, every occurrence of such $X$ is $W$  occurs as a factor $c\ov c X$. This is due to the definition of $\cX_{a}$. 

After renaming we have either 
$\sig(X)\in ((c \ov c)^4)^*$ or $\sig(X)\in ((c \ov c)^4)^* v ((c \ov c)^4)^*$ with $1 \neq v \notin c\ov cM(B',\theta,\mu)\cup M(B',\theta,\mu)c\ov c$.
\item 
Remove all variables with $\sig(X)= 1$ and rename the current vertex
as $(W,B',\cX,\theta,\mu)$. 
\item 
Call the following  while loop, called 
{\em $\Lam_a$-compression}. Repeat the loop until $\Lam_a = \os 0$. 
In the beginning of the loop we have  $\Lam_a \neq \es$ and we don't have $0\in \Lam_a$, but the number $0$ sneaks in. Actually we will have that $\Lam_a = \os 0$ \IFF 
 and there are no more factors $c\ov c$.
Note that there at least as many subsets 
 $\os{c_{\lam},\ov{c_{\lam}}}$ available as there are numbers in $\Lam_a$. 
 During the following process we will mark some $c_{\lam}$, $\ov{c_{\lam}}$ 
 to make sure that we use each set $\os{c_{\lam},\ov{c_{\lam}}}$  only once.
 At the beginning all $c_{\lam}$ are unmarked.\\ 
 
{\bf begin $\Lam_a$-compression}\\
As above, we realize each transformation via arcs satisfying the forward property. 
We continue to denote, by default,  each current vertex as $(W,B',\cX,\theta, \mu)$ and the current solution as $(\alp,\sig)$.

The first while loop is realized as a path in $\cG_\F$ by renaming arcs with label $h(c_\lam) =c$; the second one uses \subst arcs with label 
 $h(c_\lam) =c\ov c c_\lam$. \\

Let $\Lam = \Lam_{a}$.\\
{\bf while} $\Lam \neq \os 0$ {\bf do}
\begin{enumerate}
\item {\bf while} there is $\ell \in \Lam$ such that $\ell$ is odd {\bf do}
\begin{itemize}
\item let $\ell$ be largest odd number in $\Lam$;
\item choose an unmarked letter $c_\lam$ and mark $c_\lam$ and $\ov{c_\lam}$;
\item replace every maximal occurrence of $(c\ov c)^\ell$ by $c_{\lam}\ov{c_{\lam}} (c\ov c)^{\ell-1}$;\\ 
whenever possible make the factor $c_{\lam}\ov{c_{\lam}}$ visible;\\
(we have $c_{\lam}\ov{c_{\lam}} (c\ov c)^{\ell-1} = (c\ov c)^{\ell_1}c_{\lam}\ov{c_{\lam}} (c\ov c)^{\ell_2} \in M(B',\theta,\mu)$ for all $\ell_1 + \ell_2= \ell-1$, hence $c_{\lam}\ov{c_{\lam}}$ can be chosen to be visible for an occurrence of $(c\ov c)^\ell$ unless no position of that occurrence is visible)
\item replace $\Lam$ by $(\Lam \sm \os \ell) \cup \os{\ell -1}$.
\end{itemize}
{\bf end while}

Note that now, if $(c \ov c)^\ell$ is a maximal occurrence in $\sig(W)$ then 
$\ell \in \Lam$ and all $\ell \in \Lam$  are even.

\item {\bf while} $\sig(W)$ has a maximal occurrence of some factor $c_{\lam}\ov{c_{\lam}} (c\ov c)^{\ell}$ with $\ell \equiv 2 \bmod 4$ {\bf do}
\begin{itemize}
\item let $\ell$ be largest number in $\Lam$ such that there is a maximal occurrence of the factor $c_{\lam}\ov{c_{\lam}} (c\ov c)^{\ell}$ with $\ell \equiv 2 \bmod 4$;
\item replace every occurrence of $c_{\lam}\ov{c_{\lam}} (c\ov c)^{\ell}$ by $c_{\lam}\ov{c_{\lam}} (c\ov c)^{\ell-2}$;\\
realize this transformation via some substitution arc with label $h(c_\lam) =c\ov c c_\lam$; note that $h(c_\lam\ov{c_\lam}) =c\ov c c_\lam \ov{c_\lam}c \ov c =  c_\lam \ov{c_\lam} (c\ov c)^2)\in M(B',\theta,\mu)$
\item replace $\Lam$ by $\Lam \cup \os{\ell -2}$;
\item if $\sig(W)$ does not contain any maximal occurrence of  $(c\ov c)^{\ell}$ then replace $\Lam$ by $\Lam \sm \os {\ell}$.
\end{itemize}
{\bf end while}

Note that the first two while loops did not use any \subst arc since no ``uncrossing'' was necessary. Moreover, now all $\ell \in \Lam$ are even and if there is a maximal occurrence of some factor $c_{\lam}\ov{c_{\lam}} (c\ov c)^{\ell}$ then we have $\ell \equiv 0 \bmod 4$. 
Thus dividing such value $\ell$ by two keeps this value even. 
\item Follow a substitution arc with label $h(c) =c\ov c$ in order to replace all 
maximal occurrences of factors $(c\ov c)^{\ell}$ with $\ell \in \Lam$ by 
$(c\ov c)^{\ell/2}$;\\
replace $\Lam$ by the set $\set{\lam/2}{\lam \in \Lam}$;\\
(note that $\Lam$ may have odd numbers, again)

\item Remove all variables $X$ with $\abs{\sig(X)} \leq 10$. 
\item For each of the remaining $X\in \cX_{a}$ there is a unique factorization  
$\sig(X)= uv$ such that $u \in (c\ov c)^*$ and $v \notin (c\ov c)^*M(B',\theta,\mu)$.
Using \subst arcs we may assume that $u \in (c\ov c)^{4m}$ for some $m \in \N$. 
\end{enumerate}
{\bf end while}\\
{\bf end $\Lam_a$-compression}
\item \label{num524}
Let $B= B'\sm\os{c,\ov c}$. Since $\Lam_a= \os 0$ after at the end of the $\Lam_a$-compression,  no $c$ or $\ov c$ appears in $\sig(W)$: they are all compressed into single letters $c_{\lam}$ or $\ov{c_{\lam}}$. Moreover, for  $x \in B\cup \cX$ we have 
$\theta(x) =\es$. Hence we can follow an alphabet reduction arc $(W,B',\cX,\theta,\mu) \arc \eps (W,B,\cX,\es,\mu)$. The new solution to $(W,B,\cX,\es,\mu)$ is the pair 
$(\alp',\sig)$ where $\alp'= \alp \eps$ is defined by the restriction 
of $\alp$ to the free monoid $M(B,\es,\mu')$. We rename the current vertex and its solution as
$(W,B,\cX,\es,\mu)$ and $(\alp,\sig)$.
\end{enumerate}
{\bf end $(a\ov a)$-compression}\\

Having performed one round of a non-standard block compressions, we have increased the length of $W$ by at most $\Oh(n)$.  We end up at a vertex named again as $V= (W,B,\cX,\es,\mu)$ which has a solution
$(\alp, \sig)$. The difference is that $W$ has no proper factor $a \ov a a $ with $a \in B$ anymore. 

\goodbreak
\subsection*{Repeat: 1.~standard block compression, 2.~non-standard block compression, 3.~pair compression until $\cX = \es$}\label{sec:coro}
The title of this section explains what we do. We repeat rounds 
of 1.~standard block compression, 2.~non-standard block compression, 3.~pair compression. After the non-standard block compression the word $W$ does not contain any proper factor of the form $a^\ell$ with $\ell\geq 2$ or
$a\ov a a$ where $a\in B$. Thus, at the beginning of pair compression a proper factor $u$ of $W$ of length three looks as 
$u = abc$ with either $c \notin \os{b,\ov b}$ or $a \notin \os{b,\ov b}$.

To see that, indeed, all factors $a\ov a a$ vanish,
consider for example the case that before (and after) the standard block compression 
the word $W$ contains a factor 
$ba\ov a ab$ with $a \neq b$. After that the non-standard block compression uses renaming. It changes this factor either to 
$bc\ov c ab$ (if $a \in B_+$) or to $ba\ov c c b$ (if $\ov a \in B_+$); and at the end of the non-standard block compression this factor appears  either as $bc_1\ov c_1 ab$ or as $ba\ov c_1 c_1 b$.

The pair compression we proceed according to \prref{sec:pair}, but we never compress any pair $a\ov a$. Now consider again $u \leq W$ with $u = abc$ and either $c \notin \os{b,\ov b}$ or $a \notin \os{b,\ov b}$.
By symmetry, we may assume $a \notin \os{b,\ov b}$. Now the probability 
that $ab$ is compressed is 
$\Prob{a \in L \wedge b\in R}=\Prob{a \in L}\cdot  \Prob{b\in R} =  \frac14$. It follows that the expected length 
of that factor $u$ after pair compression is at most $3 \cdot \frac34 + 2\cdot \frac14 = \frac{11}4$ which is less than $3$. Thus, we obtain a recursion of type 
$s(1) \in \Oh(n)$ and $s(i+1) \leq q \cdot s(i) + \Oh(n)$ 
for all $i \in \N$ where $q = \frac{11}{12} < 1$. Such a recursion implies 
 $s(i) \in \Oh(n)$ for all $i \in \N$. This proves 
 \prref{thm:procentral}. 
\qed

\bibliographystyle{abbrv}
\bibliography{traces}
\end{document}